\documentclass{article}
\usepackage{graphicx} 
\usepackage{color,amsmath}
\usepackage{amssymb}
\usepackage{amsthm}
\usepackage{algorithm}
\usepackage{algpseudocode}
\usepackage{graphicx}
\usepackage[a4paper, left=3cm, right=3cm, top=3cm, bottom=3cm]{geometry}
\usepackage{xcolor}
\usepackage{appendix}
\usepackage{graphicx}
\usepackage{subcaption}

\newtheorem{theorem}{Theorem}[section]
\newtheorem{lemma}{Lemma}[section]

\begin{document}
\title{Global Convergence of Policy Gradient for Entropy Regularized Linear-Quadratic Control with Multiplicative Noise }
\date{}
\author{Gabriel Diaz \thanks{ Department of Mathematics,
University of California, Berkeley,
970 Evans Hall
Berkeley, CA 94720. \\Email: gdiaz2030@berkeley.edu} 
\and 
Lucky Li \thanks{College of Computing, Data Science, and Society, University of California, Berkeley, CA 94720. \\ Email: luckyql@berkeley.edu, luckyql17@gmail.com} \and Wenhao Zhang  \thanks  {Department of Applied Mathematics, The Hong Kong Polytechnic University, Hong Kong, China. \\Email:  wen-hao.zhang@connect.polyu.hk}}

\maketitle

\begin{abstract}
    Reinforcement Learning (RL) has emerged as a powerful framework for sequential decision-making in dynamic environments, particularly when system parameters are unknown. This paper investigates RL-based control for entropy-regularized linear-quadratic (LQ) control problems with  multiplicative noise over an infinite time horizon. First, we adapt the regularized policy gradient (RPG) algorithm to stochastic optimal control settings, proving that despite the non-convexity of the problem, RPG converges globally under conditions of gradient domination and almost-smoothness. Second, based on zero-order optimization approach, we introduce a novel model free RL algorithm: Sample-based regularized policy gradient (SB-RPG). SB-RPG operates without knowledge of system parameters yet still retains strong theoretical guarantees of global convergence. Our model leverages entropy regularization to address the exploration versus exploitation trade-off inherent in RL. Numerical simulations validate the theoretical results and demonstrate the efficiency of SB-RPG in unknown-parameters environments.
\end{abstract}

\section{Introduction}
Reinforcement learning (RL) is a subfield of machine learning that focuses on training agents to make sequential decisions by interacting with dynamic environments. Unlike supervised learning, which relies on labeled datasets, RL agents learn through trial and error, guided by a reward signal that quantifies the desirability of their actions \cite{Sutton2018RL}. The primary goal is to discover an optimal policy—a mapping from states to actions which maximizes the cumulative long-term rewards. In recent years, RL has achieved remarkable success, attaining human-level performance in domains such as game playing \cite{Silver2017}, robotics \cite{Kobor2013}, and autonomous driving \cite{Kiran2022}.

Optimal control theory seeks to design control policies that maximize a predefined performance criterion for dynamic systems. RL and optimal control are naturally aligned in their fundamental principles as both approaches incorporate dynamical decision-making considerations. However, traditional optimal control methods require complete knowledge of all environmental parameters , which is often infeasible in real-world applications.  RL-based control methods have thus gained prominence for scenarios with unknown system parameters. LQ
control problem, as one of the most fundamental problems in the control theory, has attracted considerable attention and has been extensively studied in the RL based control literature. For example, in continues time settings, Wang and Zhou    \cite{Wang2020} adopt the RL method to solve mean–variance portfolio problem,  Wang et, al. \cite{Wang2020b} carry out a complete theoretical analysis of RL based LQ control problem, 
Li et, al. \cite{Na2022tac, Na2023} employ a policy iteration RL approach to investigate LQ and mean-field LQ control problems over an infinite horizon.

Policy gradient methods \cite{Sutton1999} form a prominent class of RL algorithms that directly parameterize the policy and optimize it via gradient descent. Despite their simplicity and broad applicability, policy gradient methods face a non-convex optimization landscape even for basic LQ control problems \cite{Fazel2018Global}. Understanding their convergence properties is thus a central research topic. Fazel et al. \cite{Fazel2018Global} established global convergence for policy gradient in deterministic, infinite-horizon LQ settings. Subsequent works extended these results to LQ problems with multiplicative noise\cite{Gravell2019robust, Gravell2020learning}, additive and multiplicative noise \cite{Lai2023}, and finite-horizon settings \cite{Hambly2020NoisyLQR}. Typically, these studies first analyze global convergence under known system parameters (model-based RL), then extend to unknown parameter settings (model-free RL). Hu et al.\cite{Hu2023} provide a comprehensive survey of recent theoretical advances in this area.

A key challenge in RL is the exploration-exploitation trade-off: agents must balance exploiting known information to maximize immediate reward with exploring new actions to discover potentially better strategies. In most cases, exploration is highly resource-intensive, therefore, numerous solutions are proposed to address exploration-exploitation trade-off. The trade off between exploration and exploitation has been thoroughly studied for the LQR. For example, in \cite{Abeille2017ThompsonLQ} and is improved upon in \cite{Abeille2018ImprovedTS}.  In the discrete action space setting, $\epsilon$-greedy policy \cite{Watkins1992} and Botzman (Softmax) policy \cite{Sutton1990} are two effective and popular ways to balance exploration-exploitation and  numerous developments have been made based on them.  In addition to these,  recent research has introduced entropy-regularized RL formulations. This approach explicitly integrates exploration into the optimization objective by including entropy as a regularization term, thereby imposing a trade-off weight on the entropy of the exploration strategy .  
Ahmed et al. \cite{Ahmed2019} demonstrate that, even when the exact gradient is available, policy optimization remains challenging because of the complex geometry of the objective function. Moreover, employing policies with higher entropy can smooth the optimization landscape, facilitating connections between local optima. 
  Neu et, al. \cite{Neu2017} propose a general framework for entropy-regularized average-reward RL in Markov decision processes. It is noteworthy that while entropy-regularization has been quite useful it comes with the caveat that there are many other ways to promote policy exploration nor is entropy-regularization always effective as seen in \cite{Leffler2007Efficient}. 
  
 In the context of entropy-regularized RL for LQ problems, most research has focused on actor-critic methods \cite{Wang2020, Wang2020b, Na2023}, with relatively few studies addressing policy gradient approaches. Notably, Michael et al. \cite{Michael2024} study the global linear convergence of policy gradient methods for finite-horizon continuous-time entropy-regularized LQ control problems. In  discrete time settings, Guo et, al. \cite{Guo2023fast}  proposed and analyzed two new policy gradient based RL method for entropy-regularized LQ problems: regularized policy gradient (RPG) and iterative policy optimization (IPO) and proved their fast convergence given exact model parameters(i.e., model based). However, they only considered additive noise and conducted their analysis solely in the model-based setting, which limits the applicability of their methods. Multiplicative noise models may be produce more robust policies, they are still more complex than the typical additive model which may result in slower convergence. To the best of our knowledge, addressing the case of multiplicative noise in both model-based and model-free settings  still remains largely unexplored.

This paper makes two primary  contributions: First, 
we extend the RPG \cite{Guo2023fast} method to stochastic LQ scenarios with multiplicative noise, demonstrating that RPG achieves global convergence despite the non-convexity of the problem, thanks to gradient domination and almost-smoothness properties. This extension enhances robustness and broadens the applicability of RPG. Second, and more importantly, employing the zero order optimization technique, we propose sample based regularized policy gradient (SB-RPG) method and rigorously prove its global convergence properties. SB-RPG does not require knowledge of the specific system parameters values, enabling effective control in parameter-unknown environments where RPG cannot be directly applied. Our theoretical results are supported by numerical simulations.

 The rest of sections are organized as follows. In section \ref{section: formulation} we formulate the optimal control problem and transform it into a optimization problem. It is natural to consider first order method to cope with optimization problem so we give the explicit form of gradient with respect to optimization variables. In section \ref{section: model based}, we consider regularized policy gradient method proposed in \cite{Guo2023fast} and provide the guarantee of global convergence. In section \ref{section: model free}, we consider the case where all the parameters are unknown, we proposed sample based regularized policy gradient (SB-RPG) to cope with this situation. In section \ref{section: numerical}, we provide numerical experiments showing the effectiveness of our algorithm. For the sake of brevity, only the proofs of the main theorems are presented in the main text, while all detailed proofs of the lemmas are relegated to the Appendix.

\section{Formulation}\label{section: formulation}
In this section, we clearly formulate the stochastic optimal control problem over an infinite time horizon with a constant discounted rate $\gamma$ and derived the optimal feedback control policy in $\eqref{eq: pi*}$. Inspired by the structure of feedback control policy, we linearly parameterize our policy. By doing so, we transformed the optimal control problem into an finite dimensional (i.e., $n$-dimensional) optimization problem. Since optimization naturally involves consideration of first-order derivatives, we provide the explicit form of the first-order derivative. We adopt standard mathematical notation throughout this paper. For any matrix \( Z \in \mathbb{R}^{n \times m} \), we denote \( \|Z\| \) as the spectral norm of \( Z \), \( \|Z\|_F \) as the Frobenius norm of \( Z \), \( \sigma_{\min}(Z) \) and \( \sigma_{\max}(Z) \) as the minimum and maximum singular values of \( Z \), respectively, and \( Z^\top  \) as the transpose of \( Z \).

Consider the following discrete-time exploratory stochastic LQ control system in the infinite time horizon:
\begin{align}\label{eq:state} 
x_{t+1} = Ax_t + Bu_t + Cx_tw_t^x + u_t^\top D w_t^u,
\end{align}
where $A, C, x_t, w_t^x \in \mathbb{R}, ~ B\in \mathbb{R}^{1\times n}, u, w_t^u\in \mathbb{R}^{n\times 1}, ~ D \in \mathbb{R}^{n\times n}, $ and $w_t^x,w_t^u$ are white noise, which are distributed as follows 
\begin{align*}
&\mathbb{E} [w^x_t] = 0,  ~~\mathbb{E}[(w_t^x)^2] = 1,  ~~ \mathbb{E}[w^u_tx_t^x] = \mathbf{0}_{n\times 1}, \\
&\mathbb{E}[w^u_t] = \mathbf{0}_{n\times 1}, ~~ \mathbb{E}[(w_t^u)^\top w_t^u] = I_{n\times n}.
\end{align*}
Unlike traditional optimal control problems that focus solely on deterministic control, we incorporate entropy regularization and consider stationary randomized Markovian policies. This approach enables us to effectively address the exploration problem in Reinforcement Learning.
Specifically, we define the set of admissible policies as $\Pi := \{\pi : \mathcal{X} \to \mathcal{P}(\mathcal{U})\}$, where \(\mathcal{X}\) denotes the state space, $\mathcal{U} $the action space, and \(\mathcal{P}(\mathcal{U})\) the set of probability measures over $\mathcal{U}$. Each admissible policy \(\pi \in \Pi\) assigns to every state \(x \in \mathcal{X}\) a probability distribution over actions in \(\mathcal{U}\).

For any given policy \(\pi \in \Pi\), the associated Shannon entropy is defined as
\[
\mathcal H(\pi(\cdot|x)) := - \int_{\mathcal{U}} \pi(u|x) \log \pi(u|x) \, du,
\]
which measures the uncertainty or information gain from exploring the environment. By incorporating this entropy term as a regularization component in the objective function, we encourage the policy to gather information about the unknown environment and to promote exploration. The objective functional then takes the following form:
\begin{equation}\label{eqn: obj initial}
     \min_{\pi \in \Pi} \, \mathbb{E}_{x \sim \mathcal{D}}[J(x)], 
\end{equation}
where $ \Pi$ is the admissible policy set and 
\begin{equation}\label{eqn: value func}
 J_\pi(x) := \mathbb{E}_\pi\left.\left[\sum_{t=0}^\infty \gamma^t \left( Q x_t^2 + u_t^{\top} R u_t - \tau \mathcal H(\pi(\cdot|x_t)) \right) \right| x_0 = x\right].
\end{equation}
Now that we have defined the exploratory stochastic LQ problem, we present the following theorem, which provides the optimal policy, the optimal objective value, and the corresponding Algebraic Riccati Equation (ARE).

\begin{theorem}[Optimal value function and optimal control]\label{thm: optimal val func and policy}
    The optimal value function $J^*:\mathcal{X} \rightarrow \mathbb{R}$ in can be expressed as $J^*(x) = Px^2 + q$ with $P$ satisfying the following Algebraic Riccati Equation (ARE)
    \begin{equation}\label{eqn:P}
         \begin{aligned}
         P &= Q+ \gamma P(A^2 + C^2) - (\gamma AP)^2 B(R + \gamma P(B^\top B+D^\top D))^{-1}B^\top ,\\
        q &=\frac{ \text{Tr}(\Sigma^* R)+\gamma P\text{Tr}(\Sigma^*(B^\top  B+D^\top  D)) - \frac{\tau}{2}(k+\log ((2 \pi)^k \operatorname{det} \Sigma^*)}{1-\gamma}\notag,    
    \end{aligned}
    \end{equation}
where
 \begin{equation}\label{eqn: K *}
 K^* = \gamma (R+\gamma P (B^\top  B+D^\top D))^{-1}A P B^\top  , \quad     \Sigma^* = \frac{\tau}{2} (  R+\gamma P (B^\top  B+D^\top D))^{-1},
 \end{equation}
    for any $x \in \mathcal{X}$, the corresponding optimal policy for system \eqref{eq:state} and objective functional \eqref{eqn: obj initial} is : 
    \begin{equation}\label{eq: pi*}
           \pi^* = \mathcal{N}(-K^*x,\Sigma^*). 
    \end{equation}
\end{theorem}
The proof of Theorem \ref{thm: optimal val func and policy} relies on the following lemma, which establishes the optimal solution for the one-step reward function in the presence of entropy regularization. This lemma provide the necessary foundation for deriving both the optimal policy and the corresponding value function under entropy-regularized rewards. The proof of Lemma \ref{lemma: solve regularzied QP} is provided in Section 8.1 of \cite{Guo2023fast}.
 
 \begin{lemma}\label{lemma: solve regularzied QP} For any given symmetric positive definite matrix $M \in \mathbb{R}^{k \times k}$ and vector $b \in \mathbb{R}^{k}$, the optimal solution $p^* \in \mathcal{P}( \mathcal{U})$ to the following optimization problem is a multivariate Gaussian distribution with  covariance $\frac{\tau}{2}M^{-1}$ and mean $-\frac{1}{2}M^{-1}b$: 
\begin{eqnarray*}
         &\min_{p \in \mathcal{P}(\mathcal{U})} &\mathbb{E}_{u \sim p(\cdot)}
         \left[
            u^{T} M u + b^{T} u + \tau \log p (u)
        \right],  \\
        &\text{\rm subject to}  
 &\int_{\mathcal{U}} p(u)du=1, \\
                  &                &p(u) \geq 0, \quad \forall u \in \mathcal{U}.
\end{eqnarray*}
 
\end{lemma}

\begin{proof} (of Theorem \ref{thm: optimal val func and policy}).
By definition of $J^*$ in \eqref{eqn: obj initial},
\begin{equation}\label{eqn: V^*}
    J^*(x) = \min_{\pi\in \Pi} \mathbb{E}_\pi
        \Bigl\{Q x^2 + u^{T} R u + \tau \log(\pi(u|x))
        + \gamma  J^* ((A +w_t^xC)x_t + (B+ w_t^uD)u_t) \Bigr\},
\end{equation}
where the expectation is taken with respect to $u\sim \pi(\cdot | x)$ and  the noise terms $w_t^u$ and $w_t^x$, with mean $0$ and covariance $I_{n\times n}$.
   Stipulating 
   \begin{equation}\label{eqn: V* stipulate}
       J^*(x) = P x^2 + q
   \end{equation}
   for a positive $P,q\in \mathbb{R}$ and plugging into \eqref{eqn: V^*},
   we can obtain the optimal value function with dynamic programming principle:
    \begin{align}
    J^*(x)&=  Q x^2 \notag  + \min_{\pi} \mathbb{E}_\pi 
         \Bigl\{ u^{T} R u +\tau \log(\pi(u|x)) \\& \quad +
          \gamma  \left[ P ((A +w^xC)x + (B+ w^uD)u)^2  + q \right] \Bigr\}\notag \\
    &=  (Q+\gamma P(A^2+C^2))x^2 + \gamma q \notag \\
     &\quad + \min_{\pi} \mathbb{E}_\pi \Bigl\{  u^\top  (R+\gamma P (B^\top  B+D^\top  D)) u+ \tau \log(\pi(u|x)) +  2\gamma A P x B u  \Bigr\}\notag.\label{eqn:value_func}
    \end{align}
Now apply Lemma \ref{lemma: solve regularzied QP} to \eqref{eqn: value func} with $M = R+\gamma P (B^\top  B+D^\top  D)$ and $b = 2 \gamma A P x B^\top $, we can get the optimal policy at state $x$:
\begin{eqnarray} \label{eqn: pi* tmp}
    \pi^*(\cdot | x) & =& \mathcal{N}\left(-(R+\gamma P (B^\top  B+D^\top D))^{-1}\gamma A P x B^\top , \frac{\tau}{2} (  R+\gamma P (B^\top  B+D^\top D))^{-1}\right)\notag\\
    &=&\mathcal{N}\left(-K^* x, \Sigma^*\right),
\end{eqnarray}
where $K^*,\Sigma^*$ are defined in \eqref{eqn: pi* tmp}.
To derive the associated optimal value function, we first calculate the negative entropy of policy $\pi^*$ at any state $x\in\mathcal{X}$:
\begin{equation}\label{eqn:neg_gua_entropy}
    \mathbb{E}_{\pi^*} [\log(\pi^*(u | x))] = \int_{\mathcal{A}} \log(\pi^*(u | x)) \pi^*(u | x)du = - \frac{1}{2}\left(k+\log \left((2 \pi)^k \operatorname{det} \Sigma^*\right)\right).
\end{equation}
Plug \eqref{eqn: pi* tmp} and \eqref{eqn:neg_gua_entropy} into \eqref{eqn: value func} to get
\begin{align*}
    J^*(x) &=  (Q+\gamma P(A^2+C^2))x^2 + \gamma q \notag \\
     &\quad + \min_{\pi} \mathbb{E}_\pi \Bigl\{  u^\top  (R+\gamma P (B^\top  B+D^\top  D)) u+ \tau \log(\pi(u|x)) +  2\gamma A P x B u  \Bigr\}\\
&=(Q+\gamma P(A^2+C^2))x^2 + \gamma q - \frac{\tau}{2}(k+\log ((2 \pi)^k \operatorname{det} \Sigma^*)\notag\\
    &\quad + Tr(\Sigma^* (R+\gamma P (B^\top  B+D^\top  D))\\
     &\quad + (K^*x)^\top \ (R+\gamma P (B^\top  B+D^\top  D)) (K^*x)- {2\gamma A P x B K^*x}\\
    & = x^2[ Q+\gamma P(A^2+C^2 -ABK^*)]\\
    & \quad - \frac{\tau}{2}(k+\log ((2 \pi)^k \operatorname{det} \Sigma^*) + \gamma q +
Tr(\Sigma^* (R+\gamma P (B^\top  B+D^\top  D))\notag.
\end{align*}

Combining this with \eqref{eqn: V* stipulate}, the proof is completed.
\end{proof}


Inspired by the form of the optimal control function, we can linearly parameterize our policy and transform the optimal control problem to the optimization problem by the parameters $(K,\Sigma)$. By doing so, the policy can be formulated as $\pi(u|x) = \mathcal{N}(-Kx,\Sigma)$. Then admissible policy set for $(K,\Sigma)$ is defined as $\Omega = \{K \in \mathbb{R}^n, \Sigma\in \mathbb{R}^{n\times n}:\gamma V_K<1, \Sigma\succ0,\Sigma^\top  = \Sigma\}$. For simplicity of notation, we define  the cost of system given the deterministic initial state $x_0$ as $C_{K,\Sigma}(x_0) {:=}\mathbb{E}\left.\big[\sum_{t=0}^\infty \gamma^t \left( Q x_t^2 + u_t^{T} R u_t + \tau \log \pi (u_t|x_t) \right) \right| u_t \sim \mathcal{N}(-Kx_t, \Sigma)\big]$.
\begin{lemma}[Optimization formulation]\label{lemma: optimization}
The optimal control problem consider in Theorem \ref{thm: optimal val func and policy} can be written as follows:
\begin{align}\label{eqn: obj optimization}
     \min_{(K,\Sigma)\in\Omega}  f(K,\Sigma) = \mathbb{E}_{x_0\sim\mathcal{D}}[C_{K,\Sigma}(x_0)] = P_K\mu + q_{K,\Sigma},
\end{align}
 where $\pi(u | x)  = \mathcal{N}\left(-K x, \Sigma\right)$, $\mu = \mathbb{E}_{x_0\sim \mathcal{D}}x^2_0$ and $x_t$ subject to the dynamics of system in (\ref{eq:state}) and $P_K,q_{K,\Sigma}$ satisfy the following functions:
\begin{equation}\label{eq: P_k = ...}
P_K = Q+ K^\top RK + \gamma P_K(A^2+C^2 + K^\top (B^\top B+D^\top D)K - 2ABK),
\end{equation}
\begin{equation*}
q_{K,\Sigma} = \frac{ \text{Tr}(\Sigma (R + \gamma P_K (B^\top B+D^\top D)) - \frac{\tau}{2} (n + \log((2\pi)^n|\Sigma|)))}{1-\gamma}.
\end{equation*}
\end{lemma}

It is noteworthy that the optimal policy for the optimization formulation (in Lemma \ref{lemma: optimization}) and the optimal control problem (in Theorem \ref{thm: optimal val func and policy}) are identical. In other words, when $K = K^*$ and $\Sigma = \Sigma^*$ , the problem attains its optimal solution and  $P_K$ satisfy the ARE in \eqref{eqn:P}. The following lemma provides the explicit form of the first-order derivative of cost function with respect to $K$ and $\Sigma$.

\begin{lemma}[Explicit form of $\nabla _Kf(K,\Sigma)$ and $\nabla _\Sigma f(K,\Sigma)$]\label{lemma: explicit form} Assume that $\gamma <1$, it holds that
$$
\nabla_K f(K,\Sigma) 
 = E_K S_K, $$
$$
\nabla_{\Sigma} f(K,\Sigma) = (1- \gamma)^{-1} \left( (R+ \gamma P_K(B^\top B+D^\top D))^\top  - \frac{\tau}{2} \Sigma^{-1}\right),
$$
where $E_K = 2RK+2 \gamma P_K[(B^\top B+D^\top D)K - AB^\top ]$, $S_K =  \sum_{t = 0}^{\infty} \mathbb{E}x_t^2$.
\end{lemma}
%

\section{Global Convergence of Regularized Policy Gradient}\label{section: model based}

It is natural to utilize a gradient descent method for addressing the optimization problem. In \cite{Guo2023fast}, the Regularized Policy Gradient (RPG) algorithm was introduced and shown to achieve global optimality in the context of noisy linear quadratic problems. In this chapter, we extend these results by demonstrating that RPG remains globally optimal for stochastic linear quadratic problems with multiplicative noise.

Consider RPG with following updating rules with a fixed step size $\eta_1$ and $\eta_2$:
\begin{align}
    K &\leftarrow K - \eta_1 \frac{\nabla_K f(K,\Sigma)}{S_{K,\Sigma}}\label{eqn:updating 1},\\
        \Sigma &\leftarrow \Sigma - \eta_2 \Sigma \nabla_\Sigma f(K,\Sigma)\Sigma\label{eqn:updating 2} .
\end{align}
By Lemma \ref{lemma: explicit form}, the above update can be written as
\begin{align*}
    K &\leftarrow K - \eta_1 E_K,\\
        \Sigma &\leftarrow \Sigma - \frac{\eta_2}{1- \gamma} \Sigma \left( (R+ \gamma P_K(B^\top B+D^\top D))^\top  - \frac{\tau}{2} \Sigma^{-1}\right) \Sigma .
\end{align*}

Before proving Global Convergence of Gradient Methods, we first introduce the following lemmas, which establish the gradient dominance condition and the smoothness property of the value function, both of which play a crucial role in the subsequent proofs.
\begin{lemma}[Gradient Domination of $f(K,\Sigma)$]\label{lemma: gradient domination}  Let $(K^*,\Sigma^*) \in \Omega$ be an global optimal policy. Assume that $(K,\Sigma)\in \Omega$ and  $\mu  >0$. Then we have
\begin{equation*}
\lambda_1E_K^\top E_K\leq f(K,\Sigma) - f(K^*,\Sigma^*)   \leq \lambda_2  \nabla_{K}f^\top (K,\Sigma) \nabla_{K}f(K,\Sigma)  +\frac{(1-\gamma)\text{Tr}[(\nabla_{\Sigma} C_{K,\Sigma}(x_0))^2]}{\sigma_{min}(R)},
\end{equation*}
where $ \lambda_1 = \frac{\mu}{ \|R + \gamma P_K(B^\top B+D^\top D)\|} $ and $\lambda_2=\frac{1}{ \mu \sigma_{min}(R)}$.
\end{lemma}

\begin{lemma}[Gradient Norm Bounds]\label{lemma: norm bounds} The gradient of $f(K,\Sigma)$ have the following bounds,
    $$\| \nabla_K f(K, \Sigma)\| \leq \overline{\| \nabla_Kf(K, \Sigma) \|} := \frac{f(K, \Sigma)}{Q} \sqrt{\frac{f(K, \Sigma)-f(K^*, \Sigma^*)}{\lambda_1}},$$
    and 
    $$\| \nabla_{\Sigma}f(K, \Sigma)\| \leq \overline{\| \nabla_{\Sigma}f(K, \Sigma)\|}:= (1-\gamma)^{-1} \big[ \| R + \gamma P_K(B^\top B+D^\top D) \| + \frac{\tau}{2\sigma_{min}(\Sigma)} \big]. $$
\end{lemma}

Now we have proved $f(K,\Sigma)$ is gradient dominated. If $f(K,\Sigma)$ is smooth and gradient dominated, then the gradient descent methods will convergence to the global optimal at a linear rate. Unfortunately,  $f(K,\Sigma)$ cannot satisfy the smoothness condition; this is due to $f(K,\Sigma) = \infty $ when $\gamma V_K \geq 1$. We consider the case where the policy $(K,\Sigma)$ is not too close to the boundary, the objective satisfies an almost smoothness condition as follows:
\begin{lemma}[``Almost" smoothness of $f(K,\Sigma)$]\label{lemma: almost smooth} Fix $0<a<1$ and define $m = \frac{log(a)-a+1}{(a-1)^2}$,  any $\Sigma  $ and $\Sigma'$ satisfies $aI\prec \Sigma \prec I$ and $aI\prec \Sigma' \prec I$, we have
\begin{align*}
    &f(K',\Sigma') - f(K,\Sigma)\\
    &= S_K [(K' -K)^\top (R + \gamma P_K(B^\top B+D^\top D))(K' -K)+ 2(K'-K)^\top E_K] \notag\\
    & \quad  +
    q_{K,\Sigma'} - q_{K,\Sigma}\\
    &\leq S_K [(K' -K)^\top (R + \gamma P_K(B^\top B+D^\top D))(K' -K)+ 2(K'-K)^\top E_K] \notag\\
    & \quad  +\frac{Tr\left( ((R+ \gamma P_K(B^\top B+D^\top D)) - \frac{\tau}{2} \Sigma^{-1})(\Sigma'-\Sigma)\right)}{(1- \gamma)}+ \frac{\tau m}{2(1-\gamma)} Tr((\Sigma^{-1}\Sigma'- I)^2).
\end{align*}
\end{lemma}

From the above lemmas, we have established that $f(K,\Sigma)$ possesses the properties of gradient domination and is ``almost" smooth. These properties make it possible to prove global convergence. Now we need to show that one step update guarantees a decrease in $f(K,\Sigma)$. To this end, we first prove that the update of $\Sigma$ is bounded.
\begin{lemma}[Boundedness of update $\Sigma$]\label{lemma: bound of sigma}
Let $(K,\Sigma) \in \Omega$ be given such that $0\prec \Sigma \preceq I$. Assume $\tau \in (0,2\sigma_{min}(R))$. Fix $a\in (0, \text{min} \{ \frac{\tau}{2\|R + \gamma P_K (B^\top B+D^\top D)\|}, \sigma_{min}(\Sigma)\})$ with $\eta_2 \leq \frac{2(1-\gamma) a^2}{\tau}$. Update of $\Sigma$ will have $ aI \prec \Sigma' \prec I$. 
\end{lemma}
The boundedness of the update to $\Sigma$ ensures that the cost function remains well-defined along the trajectory during the execution of RPG. Furthermore, we need to show that one step update guarantees a decrease in $f(K,\Sigma)$. 
\begin{lemma}[Contraction of RPG]\label{lemma: contraction of RPG} Let $(K,\Sigma) \in \Omega$ be given such that $0\prec \Sigma \preceq I$. Assume $\tau \in (0,2\sigma_{min}(R))$. Fix $a\in (0, \text{min} \{ \frac{\tau}{2\|R + \gamma P_K (B^\top B+D^\top D)\|}, \sigma_{min}(\Sigma)\})$. For $\eta_1 \leq \frac{1}{\|R + \gamma P_K(B^\top B+D^\top D)\|}$ and $
\eta_2 \leq \frac{2(1-\gamma)a^2}{\tau}$, and $0<\phi = min\{ \eta_1 \mu \frac{\sigma_{min}{(R)}}{S_{K^*,\Sigma^*}}, \frac{\eta_2 \sigma_{min} (R)}{2(1-\gamma)}\}<1,$
$$f(K', \Sigma') - f(K^*, \Sigma^*) \leq (1-\phi) (f(K,\Sigma) - f(K^*, \Sigma^*)).$$
\end{lemma}

\begin{lemma}[Lower bound of $f(K,\Sigma)$] \label{lemma: lower bodun of f}
For any $(K,\Sigma) \in \Omega$, $ f(K,\Sigma)$ has the following lower bound:
$$
f(K,\Sigma) \geq \mu P_K + \frac{\tau k}{2(1-\gamma)}log(\frac{\sigma_{min}(R)}{\pi \tau}).
$$   
\end{lemma}
With the above lemmas established, we are now ready to prove the following theorem.
\begin{theorem}[Global convergence of RPG] \label{theorem: RPG} Given $\tau\in(0,2\sigma(R)]$, $\epsilon\in(0,1)$ take $(K,\Sigma) \in \Omega$ such that $\Sigma \preceq I$. For
$$
\eta_1 = \min \bigg\{ \frac{1}{R + \frac{\gamma}{\mu} \|B^\top B+D^\top D\| \left(f(K) -\frac{\tau K}{2(1-\gamma)}log(\frac{\sigma_{min}(R)}{\pi \tau})\right)}, \frac{2}{\tau \sigma_{min}(\Sigma)}\bigg\}, 
$$
$\eta_2 = 2\tau(1-\gamma)\eta_1^2$, and for 
$$
N \geq \max \left \{ \frac{\|S_{K^*, \Sigma^*}\| }{2\mu \eta_1\sigma_{min}(R)},\frac{1}{\tau^2\eta_1^3 \sigma_{min}(R)}\right\}\log \frac{f(K,\Sigma) - f(K^*,\Sigma^*)}{\epsilon},
$$
the Regularized Policy Gradient (RPG) has the following performance bound:
$$
f(K^{(N)}, \Sigma^{(N)}) - f(K^*, \Sigma^*) \leq \epsilon.
$$
\end{theorem}
\begin{proof}
From lemma \ref{lemma: lower bodun of f} we have 
\begin{align*}
 \frac{1}{R + \gamma P_K(B^\top B+D^\top D)} &\geq \frac{1}{\|R\| + \gamma P_K\|B^\top B+D^\top D\|}\\
 &\geq \frac{1}{R + \frac{\gamma}{\mu} \|B^\top B+D^\top D\| \left(f(K) -\frac{\tau k}{2(1-\gamma)}log(\frac{\sigma_{min}(R)}{\pi \tau})\right) }\\
 & \geq \eta_1.
\end{align*}

Define $a = \tau \eta_1 \leq \frac{\tau}{\|R\| + \gamma P_K\|B^\top B+D^\top D\|} $. We will prove this theorem by induction. 
At $t = 0$,
we have $\eta_1 \leq \frac{1}{R + \gamma P_K(B^\top B+D^\top D)}, \eta_2 = 2\tau(1-\gamma)\eta_1^2 \leq \frac{2(1-\gamma)a^2}{\tau}$. Then we can apply lemma \ref{lemma: contraction of RPG} such that  $$ f(K^{(1)}, \Sigma^{(1)}) - f(K^*, \Sigma^*) \leq (1-\phi) (f(K,\Sigma) - f(K^*, \Sigma^*)),$$
and $aI\preceq\Sigma^{(1)} \preceq I$, where $\phi$ is defined in Lemma \ref{lemma: contraction of RPG}.

Assume the theorem holds at time $t$, then we have $f(K^{(n)},\Sigma^{(n)}) \leq f(K^{(n-1)},\Sigma^{(n-1)}) \leq f(K,\Sigma)$, and $aI\prec \Sigma^{(t)} \prec I$.
Then we have
\begin{align*}
\eta_1 &\leq \frac{1}{R + \frac{\gamma}{\mu} \|B^\top B+D^\top D\| \left(f(K,\Sigma) -\frac{\tau k}{2(1-\gamma)}log(\frac{\sigma_{min}(R)}{\pi \tau})\right)}\\
&\leq \frac{1}{R + \frac{\gamma}{\mu} \|B^\top B+D^\top D\| \left(f(K^{(n)},\Sigma^{(n)}) -\frac{\tau k}{2(1-\gamma)}log(\frac{\sigma_{min}(R)}{\pi \tau})\right)},
\end{align*}
and $a\leq \frac{\tau}{\|R\| + \gamma P_{K^{(n)}}\|B^\top B+D^\top D\|}$. Now Lemma \ref{lemma: contraction of RPG} can be applied such that 

$$ 
f(K^{(n+1)}, \Sigma^{(n+1)}) - f(K^*, \Sigma^*) \leq (1-\phi) (f(K,\Sigma) - f(K^{(n)}, \Sigma^{(n)})). 
$$
and $ aI\preceq\Sigma^{(n+1)}\preceq I$ The induction is complete. Finally, observe that $0<\phi \leq \frac{2\mu \eta_1\sigma_{min}(R)}{\|S_{K^*, \Sigma^*}\| } < 1$ and $\phi \leq \frac{\eta_2a\sigma_{min}(R)}{2(1-\gamma)} = \tau^2\eta_1^3 \sigma_{min}(R)$. The proof is completed.

\end{proof}

\section{Global Convergence of Sample Based Regularized Policy Gradient}\label{section: model free}
In this section, we consider a model in which all parameters, \(A, B, C, D, Q, R\), as well as the exact value of $f(K,\Sigma)$, are unknown. The only available information pertains to the form of the system dynamics and approximate values of the cost trajectories (i.e., \( \sum_{t=0}^{l-1} \gamma^t ( Q x_t^2 + u_t^{T} R u_t + \tau \log \pi (u_t|x_t) ) \), where \( l < \infty \) serves as the rollout length in the simulation environment). Employing a zero-order optimization techniques, we propose the Sample Based Regularized Policy Gradient (SB-RPG) method for our stochastic optimal control problems. This section demonstrates that, even under settings with unknown parameters, our approach achieves globally optimal solutions with high probability. The pseudocode for SB-RPG is provided in Algorithm \ref{alg:model free RPG}, where $\widehat \nabla_K, \widehat \nabla_\Sigma,$ and $ \widehat S$ denote sample based estimate of $\nabla_Kf(K,\Sigma), \nabla_\Sigma f(K,\Sigma)$ and $S_{K,\Sigma}$, respectively.
\begin{algorithm}[h]
\caption{Pseudocode code of Sample Based Regularized Policy Gradient (SB-RPG)}\label{alg:model free RPG}
\begin{algorithmic}
\State \textbf{Input:} initial policy $(K,\Sigma) \in \Omega$, updating  steps $N$, policy estimate trajectories $M$, roll out length $l$, smoothing parameters $r_1$ and $r_2$.
\For{$j = 1, \cdots N$}
		\For{$i = 1, \cdots M$}  
		\State Sample a policy $K_{i} = K+U_{i}$, where $U_{i}$ is drawn uniformly at random over $\| U_{i}\|_{F} =  r_1$.
		\State Simulate $f^{(l)}_i(K_i,\Sigma) = \sum_{t=0}^{\ell-1}\gamma^t \left( Q x_t^2 + u_t^{T} R u_t + \tau \log \pi (u_t|x_t)  \right)$  and $\displaystyle S_i^{(l)} = \sum_{t=0}^{l-1}\gamma^t x_t^2$ under policy $(K_i,\Sigma)$ for $l$ steps. 
		\EndFor
		\vspace{0.1cm}
\State Estimate: $\widehat{\nabla}_K = \frac{1}{M} \sum_{i=1}^M \frac{n}{r_1^2} f^{(l)}_i(K,\Sigma) U_{i}$, $\widehat{S} = \frac{1}{M}\sum_{i=1}^M S_i^{(l)}$
\State Update: $K \leftarrow K - \eta_1{\widehat{\nabla}_K}/{\widehat{S}}$
		\For{$i= 1, \cdots M$}  
		\State Sample a policy $\Sigma_{i} = (L+V_{i})(L+V_i)^\top$, where $V_{i}$ is drawn uniformly at random over $\| V_{i}\|_{F} =  r_2$.
		\State Simulate the cost of $f^{(l)}_i(K,\Sigma_i) = \sum_{t=0}^{l-1}\gamma^t \left( Q x_t^2 + u_t^{T} R u_t + \tau \log \pi (u_t|x_t)  \right)$  under policy $(K, \Sigma_i)$ for $ l $ steps 
		\EndFor
		\vspace{0.1cm}
\State Estimate: $\widehat{\nabla}_\Sigma = \frac{L^{-1}}{2M}\sum_{i=1}^M \frac{n(n+1)}{2r_2^2} f^{(l)}_i(K,\Sigma_i) V_{i}$
\State Update: $\Sigma \leftarrow \Sigma - \eta_2 \Sigma \widehat{\nabla}_\Sigma \Sigma$
\,
\EndFor
\end{algorithmic}
\end{algorithm}

To prove global convergence of SB-RPG, we need to  prove step by step that all sample-based estimates, under some condition on $l,M$ and $r$, 
can be $\epsilon$ close to the true value with high probability. To this end, perturbation analysis and several other technical tools are essential. We present the following lemmas, with proofs provided in the appendix. 
\begin{lemma}\label{lemma:upper bound of S_K} For any $ (K,\Sigma)\in \Omega$, $S_{K,\Sigma}$ can be written as
    $$S_{K,\Sigma}= \mu\sum_{t=0}^{\infty}(\gamma V_k)^t + \frac{Tr(\Sigma(B^\top B + D^\top D))}{1-V_K} \left[\frac{1}{1-\gamma} - \frac{1}{1-\gamma V_K} \right].$$
Furthermore, $S_{K,\Sigma}$ has the following bound 
    $$\frac{\mu}{1-\gamma V_K} \leq S_K \leq  \frac{f(K, \Sigma)-{(1-\gamma)^{-1}}\left[ Tr(\Sigma R) - \frac{\tau}{2} \left( n + log(2\pi)^n|\Sigma| \right)\right]}{Q}.$$
\end{lemma}
\begin{lemma}[Approximate $f(K,
\Sigma)$ and $S_{K,\Sigma}$ with any desired accuracy]\label{lemma: approximate fk}
 For any $K,\Sigma$ with $f(K,\Sigma)<\infty$, let $f^{(l)}(K,\Sigma) = \mathbb{E} \left[\sum_{t=0}^{l-1}\gamma^t \left( Q x_t^2 + u_t^{T} R u_t + \tau \log \pi (u_t|x_t)  \right)\right]$ and $ S^{(l)}_{K,\Sigma} =  \sum_{t=0}^{l-1} \gamma^t \mathbb{E}x_t^2$. we have

 (i) $S_{K,\Sigma} - S^{(l)}_{K,\Sigma} \leq \epsilon
$, if 
$$
l\geq \frac{\log \epsilon - \log S_{K,\Sigma}}{log \gamma}.
$$

(ii) $f(K,\Sigma) - f^{(l)}(K,\Sigma)\leq \epsilon$, if
$$
l\geq \frac{\log \epsilon - \log \left[(Q + K^\top RK)S_{K,\Sigma} + \frac{ Tr(\Sigma R) - \frac{\tau}{2}(n+ log(2\pi)^n|\Sigma|)}{1-\gamma }\right]}{log \gamma}.
$$
\end{lemma}

\begin{lemma} [$S_{K,\Sigma}$ Perturbation]\label{lemma: sk perturbation}  
If
$
\|K-K'\|\leq h_\Sigma  
$
and 
$\|\Sigma' - \Sigma\|\leq \|\Sigma\|$, then 
$$
|S_{K',\Sigma'}-S_{K,\Sigma}| \leq h_K\|K-K'\| + h_{2}\|\Sigma - \Sigma'\|,
$$
where
\begin{align*}
  h_\Sigma &= \frac{1}{2S_{K,\Sigma}^2}\frac{\mu^2 }{\sqrt{\frac{1}{2}\frac{\mu^2 \|B^\top B+D^\top D\|}{S_{K,\Sigma}^2 } + \|K^\top (B^\top B + D^\top D) - AB\|^2} + \|K^\top (B^\top B + D^\top D) - AB\|},\\
   h_2 & = \frac{\gamma Tr((B^\top B+D^\top D))}{(1-\gamma)(1- \gamma V_K)},\\
   g_\Sigma& = 2\left(\frac{1}{1-\gamma V_K } \right)^2 (2 \|K^\top (B^\top B + D^\top D) - AB\| +\|B^\top B+D^\top D\| h_\Sigma), \\
   h_K &= 2g_\Sigma \frac{(1- \gamma)\mu+ \gamma Tr(\Sigma(B^\top B+D^\top D))}{(1-\gamma)} .
\end{align*}

\end{lemma}

\begin{lemma}[$P_K$ perturbation]\label{lemma: pk perturbation} If  
$\|K'-K\| \leq \min \{h_{\Sigma},\|K\|\} $, then $$ |P_{K'} - P_K| \leq h_5 \|K' -K \|, $$
where $$h_5 = \frac{3\|K\|\|R\|}{1-\gamma V_K}+ (Q+4\|R\|\|K\|^2) g_\Sigma.$$
\end{lemma}

\begin{lemma}[$\nabla_K f(K, \Sigma)$ and $\nabla_\Sigma f(K, \Sigma)$ perturbation]\label{lemma: gradient perturbation} If 
$
\|K' - K\| \leq \min \{h_\Sigma,\|K\|\}  
$
and
$
\|\Sigma'-\Sigma\|_F\leq \min \{\frac{\sigma_{min}(\Sigma)}{2},  \|\Sigma\|\}
$, then  
$$
\|\nabla_K f(K', \Sigma') - \nabla_K f(K, \Sigma)\| \leq h_6 \|K' - K\| +h_7 \| \Sigma' - \Sigma \| ,
$$
and
$$
    \|\nabla_{\Sigma} f(K', \Sigma') - \nabla_{\Sigma} f(K, \Sigma)\| \leq h_8 \| K'-K \| + h_9\|\Sigma' - \Sigma\|,
$$
where 
\begin{align*}
h_E &= 2(\|R\| + \gamma A \cdot h_5 \|B\| +\gamma P_{K}(B^\top  B+D^\top  D) + 2\gamma \cdot h_5\|B^\top  B+D^\top  D\|\|K\|) ,\\ 
h_6 &= h_K  \sqrt{\lambda_1^{-1}|f(K, \Sigma) - f(K^*, \Sigma^*)| } + h_E |S_{K', \Sigma'}|, \quad h_9 =  \frac{\tau\sigma_{min}(\Sigma)}{4(1-\gamma)},\\
h_7 &= h_2  \sqrt{\lambda_1^{-1}|f(K, \Sigma) - f(K^*, \Sigma^*)| }, \quad 
h_8 =\frac{\gamma (B^\top B +D^\top D) }{(1-\gamma)} h_5.
\end{align*}

\end{lemma}
We define
$
f_{r_1}(K, \Sigma) := \mathbb{E}_{U\sim \mathbb{B}_{r_1}}[f(K+U, \Sigma)]
$
and 
$
f_{r_2}(K, L) := \mathbb{E}_{V \sim \mathbb{B}_{r_2}}[f(K, (L+V)(L+V)^\top]
$. where $\mathbb{B}_{r}$ denotes the uniform distribution over the points with norm $r$ (boundary of a sphere). The following lemma shows that the gradient of $f_{r_1}(K, \Sigma)$ and $f_{r_2}(K, L)$ can be estimated with an oracle for the function value.
\begin{lemma}\label{lemma: fr(K)}
$$
\nabla_K f_{r_1}(K, \Sigma) = \frac{n}{r_1^2}\mathbb{E}_{U\sim \mathbb{S}_{r_1}}[f(K+U, \Sigma)U],
$$
and
$$
\nabla_L f_{r_2}(K, L) = \frac{n(n+1)}{2r_2^2}\mathbb{E}_{V \sim \mathbb{S}_{r_2}}[f(K, (L+V)(L+V)^\top)V],
$$
where $\mathbb{S}_{r}$ denotes the uniform distribution over all points with norm at most r (the entire sphere).
\end{lemma}

The following lemmas show that $\nabla_K f(K, \Sigma)$, $\nabla_{L}f(K, L)$ and $S_{K,\Sigma}$ can be estimated with finite samples under small perturbation at any desired accuracy.
\begin{lemma}[Estimate of $\nabla_K f(K,\Sigma)$] \label{lemma: est gk}


  Given an arbitrary tolerance $\epsilon_1>0$
 and probability $\kappa_1\in(0,1)$. Let $x_t^i, u_t^i$ be $i$-th single path sampled using policy  $(K + U_i, \Sigma) \in \Omega$, where $\|U_i\|_F \leq r_1$, define
$$
\widehat{\nabla}_K := \frac{1}{M} \sum_{i=1}^M \frac{n}{r_1^2} \left[\sum_{t = 0}^{l-1} \gamma^t\left(Q(x_t^i )^2 + (u_t^i)^\top Ru_t^i  + \tau \log \pi (u_t|x_t)\right)\right] U_{i}.
$$
Assume  that (i) the distribution of the initial states implies that $\|x_0^i\| \leq \bar L$ almost surely for any $i$. (ii) the multiplicative noise are distributed such that $  \sum_{t=0}^{l-1} Q(x_t^i)^2 + (u_t^i)^\top R(u_t^i) + \tau \log\pi(u_t^i|x_t^i)  \leq \Gamma \mathbb{E}\left[ \sum_{t=0}^{l-1} Qx_t^2 + u_t^\top Ru_t +\tau \log\pi(u_t^i|x_t^i) \right]$ for any $i$. If  set 
\begin{align*}
    r_1 &\leq \frac{\epsilon_1}{2 h_6},\\
    M&\geq \max \big \{ \frac{2n}{(\epsilon_1/6)^2}( \sigma_1 + \frac{R_1\epsilon}{18\sqrt{n}})\log\big(\frac{n+1}{\sqrt{\kappa_1}}\big), \frac{2n}{(\epsilon/3)^2}(\sigma^2_{2} + \frac{R_2 \epsilon}{9x\sqrt{n}})\log(\frac{n+1}{\sqrt{\kappa_1}})\big \},\\
    l &\geq  \log (\gamma)^{-1} \left[\log\left(\frac{r_1}{n} \cdot \frac{\epsilon}{3} \right)- \log \left( 2|f(K, \Sigma)| \left(2 \| K \|^2\|R \| + \frac{1}{|Q|}\right) + |\psi|\left(1 + \frac{1}{|Q|} + \frac{1}{1-\gamma}\right) \right) \right],
\end{align*}
then
$$||\widehat{\nabla}_K - f(K,\Sigma)||_F \leq \epsilon$$ 
with  high probability (at least $1-\kappa_1$), where $\psi_1 = Tr(\Sigma R) - \frac{\tau}{2}(n + \log(2\pi)^n|\Sigma|), \sigma_1 = \big(\frac{2n}{r_1}f(K, \Sigma)\big)^2 + \big(\frac{\epsilon}{6} + \overline{\| \nabla_K f(K, \Sigma) \|}\big)^2, R_1 = \frac{2n}{r_1}f(K, \Sigma) + \frac{\epsilon}{6} + \overline{\| \nabla_K f(K, \Sigma) \|},\sigma_2 = (2\Gamma \bar L^2 f(K, \Sigma)r_1)^2 + (\frac{\epsilon}{2} + \overline{\| \nabla_K f(K, \Sigma) \|})^2,\\ R_2 =  2\Gamma L^2 f(K, \Sigma)r_1 + \frac{\epsilon}{2} + \overline{\| \nabla_K f(K, \Sigma) \|}$.

\end{lemma}

\begin{lemma}[Estimate of $\nabla_{\Sigma} f(K,\Sigma)$]\label{lemma: est gsigma}

 Given an arbitrary tolerance $\epsilon_2>0$
 and probability $\kappa_2\in(0,1)$. Let $x_t^i, u_t^i$ be $i$-th single path sampled using policy $(K, \Sigma_i) \in \Omega$, where $\Sigma_i = (L+V_i)(L+V_i)^\top$, $L$ is choloky decomposition of $\Sigma$ such that $\Sigma = LL^\top$, $\|V_i\|_F\leq r_2$,  define
\begin{align*}
	\widehat{\nabla}_\Sigma &:= \frac{1}{2}L^{-1}\frac{1}{M} \sum_{i=1}^M \frac{n(n+1)}{2r_2^2} \left[\sum_{t = 0}^{l-1} \gamma^t\left(Q(x_t^i )^2 + (u_t^i)^\top Ru_t^i  + \tau \log \pi (u_t|x_t)\right)\right] V_{i}
\end{align*}
Assume that (a)  the distribution of the initial states implies that $\|x_0^i\| \leq \bar L$ almost surely for any $i$. (b) the multiplicative noise are distributed such that
$
	\sum_{t=0}^{ l-1} Q(x_t^i)^2 + (u_t^i)^\top R(u_t^i) + \tau \log\pi(u_t|x_t) \leq \Gamma \mathbb{E}\left[ \sum_{t=0}^{l-1} Q(x_t^i)^2 + (u_t^i)^\top Ru_t^i +\tau \log\pi(u_t^i|x_t^i) \right].
$
If 
\begin{align*}
    r_2 &\leq \frac{\epsilon_2}{12h_9\|L^{-1}\|},\\
    M &\geq  \frac{2n(n+1)}{(\epsilon_2/3)^2}(\sigma_Z^2 + \frac{R_Z{\epsilon_2/3}}{3}) \log\big(\frac{n(n+1)/2+1}{1-\sqrt{1-\kappa}}\big),\\
    l &\geq  \frac{\log \frac{2r_2\epsilon_2}{6n(n+1)\|L^{-1}\|} - \log \left[(Q + K^\top RK)S_{K,\Sigma} + \frac{ Tr(\Sigma R) - \frac{\tau}{2}(n+ log(2\pi)^n|\Sigma|)}{1-\gamma }\right]}{\log \gamma},
\end{align*}
then
$$||\widehat{\nabla}_\Sigma -\nabla_\Sigma f(K,\Sigma)||_F \leq \epsilon_2$$ with high probability (at least $1-\kappa_2$), 
where $R_Z = \max\{\frac{n(n+1)}{r_2} f(K, \Sigma) + \overline {\|\nabla_L f_{r_2}(K, L)\|} + \frac{\epsilon}{12\|L^{-1}\|},\\ \frac{n(n+1)}{r_2} \Gamma L^2 f(K, \Sigma)  + \overline{\| \nabla_L f(K, \Sigma) \| }+ \frac{\epsilon_2}{3\|L^{-1}\|}\},\sigma_Z = \max\{\frac{n(n+1)}{r_2} f(K, \Sigma), \big(\frac{n(n+1)}{r_2} \Gamma \bar L^2 f(K, \Sigma)\big)^2\}$
 
\end{lemma}

\begin{lemma}[Estimate of $S_{K,\Sigma}$ under perturbation]\label{lemma: approxi S under perturbation}Given an arbitrary tolerance $\epsilon_3>0$
 and probability $\kappa_3\in(0,1)$. Let $x_t^i, u_t^i$ be a single path sampled using policy  $(K + U_i, \Sigma)\in\Omega$, where $\|U_i\|_F \leq r_3$. 
Define
$$
\widehat{S}_{K,\Sigma} := \frac{1}{M}\sum_{i=1}^{M}\sum_{t=0}^{l-1}\gamma^t(x_t^i)^2.
$$
If set  $r_3 \leq  \min\{\frac{S_{K,\Sigma}}{2h_K}, \frac{\epsilon_3}{3h_K}  ,h_\Sigma\}$, $l \geq \frac{\log \epsilon_3/3 - \log S_{K,\Sigma}/2}{log \gamma}$ and $M\geq \sqrt{\frac{3S_{K,\Sigma}}{\epsilon_3}\log\frac{n}{\kappa_3}}$,
then 
 $$|S_{K,\Sigma} - \widehat{S}_{K,\Sigma}|<\epsilon$$ with high probability (at least $1 - \kappa_3$). Furthermore, if $\epsilon_3 \leq \mu/2$, then  $\widehat S_{K,\Sigma} \geq \mu/2$ .

\end{lemma}

\begin{lemma}[$f({K,\Sigma})$ perturbation]\label{lemma: f perturbation} We have 
 $$
 |f(K,\Sigma) - f(K',\Sigma')|\leq h_{10}\|K'-K\| + h_{11}\|\Sigma'-\Sigma\|,
 $$
 if $\|K'-K\|\leq \min\{h_\Sigma, \|K\|\} $ and $\|\Sigma'-\Sigma\|\leq \|\Sigma\|$, where $h_{10} = ({2\gamma\|\Sigma\|  \|B^\top B+D^\top D\|}(1-\gamma)^{-1}  + \mu) h_5$ and $h_{11} = \frac{m\|\Sigma^{-1}\|_F}{2}+ \overline{\|\nabla_{\Sigma}q_{K,\Sigma}\|}.$ 
\end{lemma}

With the above lemmas, we can now prove the following theorem.

\begin{theorem}[Global Convergence of SB-RPG]
 Given an arbitrary tolerance $\epsilon>0$
   and probability $\kappa\in(0,1)$.  If 
   
   (i) $\eta_1, \eta_2$, $\phi$ to be equal to the  values in Theorem \ref{theorem: RPG}. 
   
   (ii) $N\geq N_{SB}$, where $N_{SB}=\frac{N_{RPG}\log(1-\phi)}{\log(1-{\phi}/{2})}$ and  $N_{RPG}$ denotes the minium update steps in Theorem \ref{theorem: RPG}. 
   
   (iii) $M$, $l$, $r_1$ and  $r_2$ satisfy the conditions in Lemma \ref{lemma: est gk}, \ref{lemma: est gsigma}, and \ref{lemma: approxi S under perturbation} when $\kappa_1 = \kappa_3 = 1-(1-\kappa)^{1/(4N_{SB})}$, $\kappa_2 = 1-(1-\kappa)^{1/(2N_{SB})}$, $\epsilon_1 = \frac{ \mu\phi \epsilon}{8\eta_1(h_{10} + h_{11})}$, $\epsilon_2 = \frac{\phi \|\Sigma\|\epsilon}{2\eta_2(h_{10} + h_{11})}$ and $\epsilon_3 = \frac{\mu^2\phi \epsilon}{8\eta_1 \overline{\|{\nabla_K f(K,\Sigma)}\|}(h_{10} + h_{11})}$. 
Then SB-RPG (in Algorithm \ref{alg:model free RPG}) will have the following performance bound after $N$ times update  
$$f(K^{(N)}, \Sigma^{(N)}) - f(K^*,\Sigma^*) < \epsilon$$
with high probability (at least $1-\kappa$).
\end{theorem}
\begin{proof}
Define $K',\Sigma'$ as the result of one step update of RPG in \eqref{eqn:updating 1} and \eqref{eqn:updating 2}.  In Lemma \ref{lemma: contraction of RPG} we have when $\eta_1$ and $\eta_2$ are chosen properly, we have 
 $$ 
 f(K', \Sigma') - f(K^*, \Sigma^*) \leq (1-\phi) (f(K,\Sigma) - f(K^*, \Sigma^*)),
 $$
Define $K'' = K-\eta_1 \frac{\widehat{\nabla}_K}{\widehat{S}_{K, \Sigma}}$ and $\Sigma'' = \Sigma - \eta_2 \Sigma \widehat{\nabla}_\Sigma \Sigma$ where $\widehat{\nabla}_K$ and $\widehat{\nabla}_\Sigma$ are defined in Lemma \ref{lemma: est gk} and Lemma \ref{lemma: est gsigma}. 
We will show that when $\nabla_K f(K,\Sigma), \nabla_\Sigma f(K,\Sigma)$ and $S_{K,\Sigma}$ are estimated accurately enough, then we have
$|f(K'', \Sigma'') - f(K', \Sigma') | \leq \frac{\epsilon}{2} \phi$
, which implies that when $f(K',\Sigma') - f(K^*, \Sigma^*) > \epsilon$, 
$$
f(K'',\Sigma'') - f(K^*, \Sigma^*) \leq (1-\frac{\phi}{2})(f(K,\Sigma) - f(K^*, \Sigma^*) )
$$
with probability $(1-\kappa)^{1/N_{SB}}$.
As we proved perturbation of $f(K,\Sigma)$ in Lemma \ref{lemma: f perturbation}, we only need to establish the following two claims, both under condition given in the theorem.

(i) $\|K'' - K'\| \leq \frac{\phi\epsilon}{2(h_{10} + h_{11})}$  with probability $(1-\kappa)^{1/(2N_{SB})}$.
\begin{align*}
    \|K''-K'\| & = \eta_1 \Bigg\| \frac{\widehat{\nabla}_K}{\widehat{S}_{K, \Sigma}} - \frac{\nabla_K f(K, \Sigma)}{S_{K, \Sigma}} \Bigg\| \notag\\
     & \leq \eta_1\frac{1}{\widehat{S}_{K,\Sigma}}  \|\nabla_K f(K,\Sigma) - \widehat{\nabla}_K \| + \eta_1 \overline{ \| {\nabla_K f(K,\Sigma)} \| } \Bigg|\frac{1}{S_{K,\Sigma}} - \frac{1}{\widehat{S}_{K,\Sigma}}  \Bigg| \notag.
\end{align*}
For the first term,  from Lemma \ref{lemma: approxi S under perturbation} and Lemma \ref{lemma: est gk} we have
\begin{align*}
 \frac{\eta_1}{\widehat{S}_{K,\Sigma}} \|\nabla_K f(K,\Sigma) - \widehat{\nabla}_K \| \leq \frac{2\eta_1}{\mu} \|\nabla_K f(K,\Sigma) - \widehat{\nabla}_K \| \leq \frac{\epsilon\phi}{4(h_{10} + h_{11})}
\end{align*}
with probability at least $ (1-\kappa)^{1/(4N_{SB})}$, as we set $\epsilon_1 =  \frac{ \mu\phi}{8\eta_1(h_{10} + h_{11})}\epsilon$ and $ 1 - \kappa_1 = (1-\kappa)^{1/(4N_{SB})}$.
For the second term, by standard matrix perturbation and Lemma \ref{lemma: approxi S under perturbation}, we have
\begin{align*}
 \eta_1\|{\nabla_K f(K,\Sigma)}\| \Bigg|\frac{1}{S_{K,\Sigma}} - \frac{1}{\widehat{S}_{K,\Sigma}} \Bigg| \leq  \eta_1 \overline{\|{\nabla_K f(K,\Sigma)}\|} \frac{2|\widehat{S}_{K,\Sigma} - S_{K,\Sigma}|}{\mu^2} \leq \frac{\epsilon\phi}{4(h_{10} + h_{11})}
\end{align*}
with probability $ (1-\kappa)^{1/(4N)}$, 
as we set $\epsilon_3  =    \frac{\mu^2\phi}{8\eta_1 \overline{\|{\nabla_K f(K,\Sigma)}\|}(h_{10} + h_{11})}\epsilon$ and  $ 1 - \kappa_3 = (1-\kappa)^{1/(4N_{SB})}$.

(ii) $\|\Sigma'' - \Sigma' \| \leq \frac{\phi\epsilon}{2(h_{10} + h_{11})}$  with probability $(1-\kappa)^{1/2N_{SB}}$.

\begin{align*}
\|\Sigma'' - \Sigma' \| & = \eta_2 \|\Sigma \nabla_\Sigma f(K,\Sigma)\Sigma - \Sigma \widehat{\nabla}_\Sigma \Sigma \|\\
& \leq  \eta_2 \|\Sigma\| \|\nabla_\Sigma f(K,\Sigma)\Sigma -  \widehat{\nabla}_\Sigma\Sigma \|\\
 & \leq \eta_2 \|\Sigma\|^2 \|\nabla_\Sigma f(K,\Sigma) -  \widehat{\nabla}_\Sigma\|.
\end{align*}
From Lemma $\ref{lemma: est gsigma}$, we have $\|\Sigma'' - \Sigma' \|\leq \frac{\epsilon\phi}{2(h_{10} + h_{11})}$ with probability $(1-\kappa)^{1/2N}$, as we set $ \epsilon_2 =  \frac{\phi}{2\|\Sigma\|^2\eta_2(h_{10} + h_{11})}\epsilon$ and  $1- \kappa_2 = (1-\kappa)^{1/(2N_{SB})}$.

Combining (i), (ii) and Lemma \ref{lemma: f perturbation}, we have 
\begin{align*}
|f(K'', \Sigma'') - f(K', \Sigma') | & \leq h_{10}\|K'' - K'\| + h_{11}\|\Sigma'' - \Sigma'\| \leq \frac{\epsilon}{2} \phi
\end{align*}
with probability $(1-\kappa)^{1/N_{SB}}$. 
Then we have 
$$f(K'', \Sigma'') - f(K^*, \Sigma^*) \leq \left( 1-\frac{\phi}{2} \right)(f(K, \Sigma)-f(K^*, \Sigma^*))$$
 with probability $(1-\kappa)^{1/N_{SB}}$ for each iteration.
Now we have proved contraction of SB-RPG, the rest of the proof remains the same as Theorem \ref{theorem: RPG}, the probability of convergence becomes $((1-\kappa)^{1/N_{SB}})^{N_{SB}} = 1-\kappa $ after $N$ times iteration.
\end{proof}

\section{Numerical Experiments}\label{section: numerical}
In this section, we provide a  numerical experiment using SB-RPG (Algorithm \ref{alg:model free RPG}), where all the parameters are unknown. 
Although SB-RPG does not require the exact values of the parameters, we need to specify the following parameters to simulate the system cost of $f(K,\Sigma)$.  We set 
$A = 0.7, B = (0.1, 0.2, 0.3), C = 0.03, Q = 0.5, \gamma = 0.5, \tau = 0.1, $
$$
D = 
\begin{pmatrix}
0.05 & 0.13& 0.12\\
0.13 &0.07& 0.10 \\
 0.12& 0.10& 0.03
\end{pmatrix} \quad
R =
\begin{pmatrix}
1& 0& 0\\
0 & 1& 0\\
0 & 0& 1\\
\end{pmatrix}
.$$

The theoretical analysis provides rather conservative limits for the step size $\eta$, number of rollouts $M$, and rollout length $l$. To ensure practicality, we determine the constant step size, number of rollouts, rollout length, and exploration radius by performing a grid search over a set of reasonable values. In simulations, we obtained the baseline optimal cost $f(K^*,\Sigma^*)$ by solving the ARE in (2) to high accuracy ($e^{-5}$ accuracy) using value iteration.

\begin{figure}[htbp]
    \centering
    \begin{subfigure}[t]{0.4\textwidth}
        \centering
        \includegraphics[width=\textwidth]{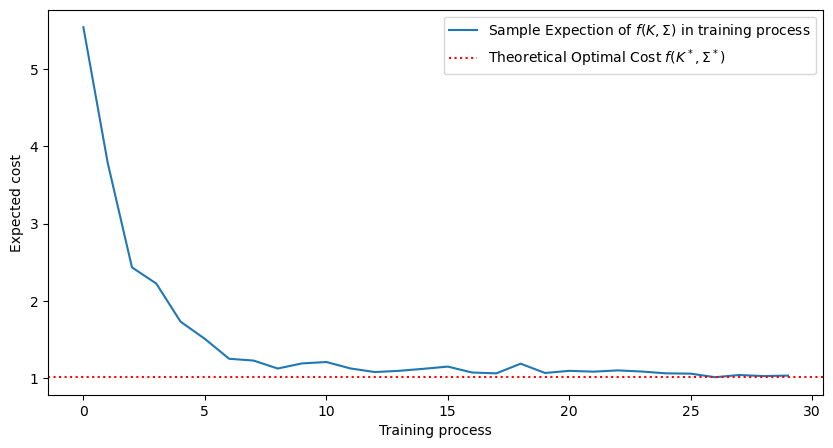}
        \caption{Expected cost of RPG}
        \label{fig:sub1}
    \end{subfigure}
    \hfill
    \begin{subfigure}[t]{0.4\textwidth}
        \centering
        \includegraphics[width=\textwidth]{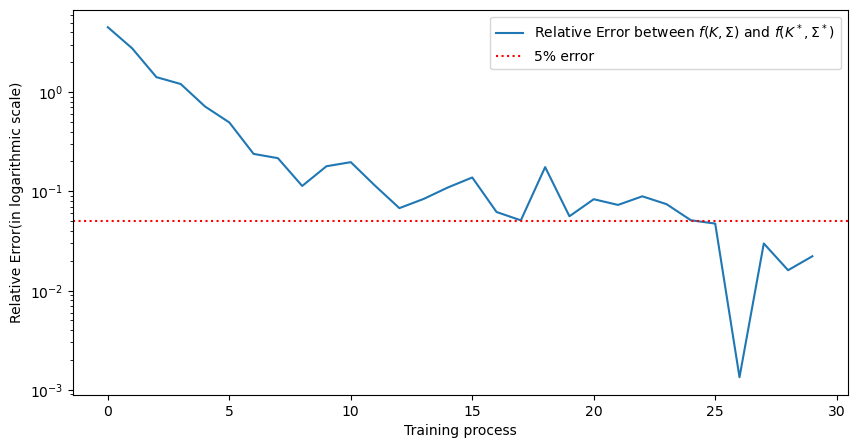}
        \caption{Relative error of $f(K,\Sigma)$ and $f(K^*,\Sigma^*)$}
        \label{fig:sub2}
    \end{subfigure}
    \caption{Expected cost of RPG}\label{fig: cost}
    \label{fig:both}
\end{figure}
In Figure \ref{fig: cost}, we compare the optimal cost $f(K^*,\Sigma^*)$ with the SB-RPG cost $f(K,\Sigma)$ throughout the training process. The left subfigure shows the absolute error $|f(K,\Sigma) - f(K^*,\Sigma^*)|$, while the right subfigure displays the relative error $\frac{|f(K,\Sigma) - f(K^*,\Sigma^*)|}{f(K^*,\Sigma^*)}$. As observed in the right subfigure, the relative error remains approximately 5\%. 
\begin{figure}[htbp]
    \centering
    \begin{subfigure}[t]{0.4\textwidth}
        \centering
        \includegraphics[width=\textwidth]{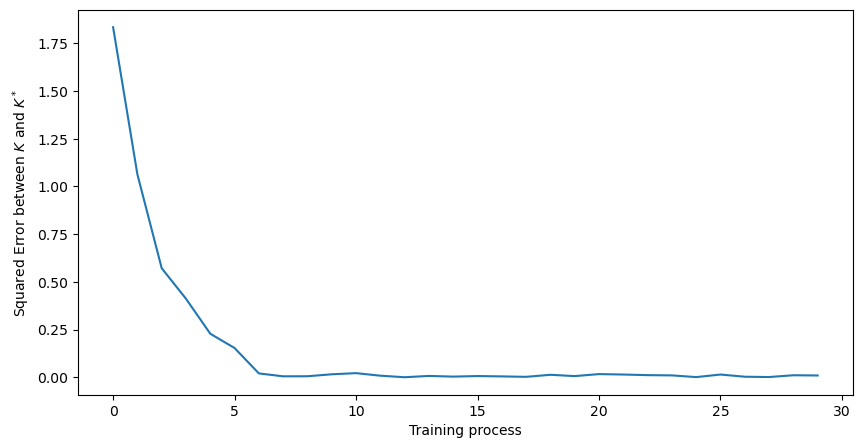}
        \caption{Squared error of $K$ and $K^*$}
        \label{fig:sub3}
    \end{subfigure}
    \hfill
    \begin{subfigure}[t]{0.4\textwidth}
        \centering
        \includegraphics[width=\textwidth]{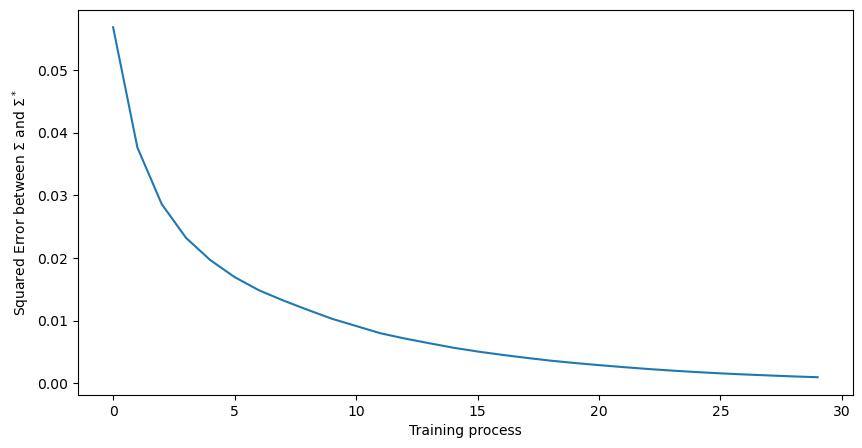}
        \caption{Squared error of $\Sigma$ and $\Sigma^*$}
        \label{fig:sub4}
    \end{subfigure}
    \caption{Squared error between optimal policy and SB-RPG}\label{fig: policy error}
    \label{fig:both}
\end{figure}

In Figure \ref{fig: policy error}, we present the relative error between the sample-based policy and the optimal policy. The left subfigure illustrates the squared error $\|K' - K\|_F^2$, and the right subfigure shows the squared error $\|\Sigma' - \Sigma\|_F^2$.

\section*{Acknowledgments}
The first two authors would like to thank the Department of Applied Mathematics at The Hong Kong Polytechnic University for the opportunity to join the research group and contribute to the meaningful work being carried out with its students.

\pagebreak

\newpage

\appendix \label{appendix}

\section{Proofs in Section \ref{section: formulation}}

\subsection{Proof of Lemma \ref{lemma: solve regularzied QP}}

The proof of Lemma \ref{lemma: solve regularzied QP} is provided in Section 8.1 of \cite{Guo2023fast}.

\subsection{Proof of Lemma \ref{lemma: optimization}}

\begin{proof} We use the similar augment with Theorem \ref{thm: optimal val func and policy}. Take the policy $(K,\Sigma)$ into objective function we have 
    \begin{align}
    J(x)&=  Q x^2 \notag  + \mathbb{E}_\pi 
         \Bigl\{ u^{T} R u +\tau \log(\pi(u|x))+
          \gamma  \left[ P ((A +w^xC)x + (B+ w^uD)u)^2  + q \right] \Bigr\}\notag \\
    &=  (Q+\gamma P(A^2+C^2))x^2 + \gamma q \notag \\
     &\quad + \mathbb{E}_\pi \Bigl\{  u^\top  (R+\gamma P (B^\top  B+D^\top  D)) u+ \tau \log(\pi(u|x)) +  2\gamma A P x B u  \Bigr\}\notag\\
     & = (Q+K^\top RK + \gamma P_K(A^2+C^2 + K^\top (B^\top B+D^\top D)K - 2ABK))x^2\notag\\
     &\quad + \gamma q_{K,\Sigma} - \frac{\tau}{2} (n + \log((2\pi)^n|\Sigma|)) + \text{Tr}(\Sigma (R + \gamma P_K (B^\top B+D^\top D)))\notag.
     \end{align}
Take the above equals to $P_K x^2 + q_{K,\Sigma}$, we have:
$$
P_K = Q+K^\top RK + \gamma P_K(A^2+C^2 + K^\top (B^\top B+D^\top D)K - 2ABK),
$$
$$
q_{K,\Sigma} = \frac{ \text{Tr}(\Sigma (R + \gamma P_K (B^\top B+D^\top D)) - \frac{\tau}{2} (n + \log((2\pi)^n|\Sigma|)))}{1-\gamma}.
$$
\end{proof}
\subsection{Proof of Lemma \ref{lemma: explicit form}}
\begin{proof}
Define $V_K = A^2+C^2 + K^\top (B^\top B+D^\top D)K - 2ABK$.
We know $f(K, \Sigma)=\mathbb{E}_{x \sim \mathcal{D}}[C(K, \Sigma)]=P_K\mathbb{E}x^2+q_{K,\Sigma}$.
$\nabla_{\Sigma}f(K,\Sigma)$ can be computed directly:
\begin{align*}
    \nabla_{\Sigma}f(K,\Sigma) & =\nabla_{\Sigma}\left(q_{K,\Sigma} \right) \notag\\
    & = \nabla_{\Sigma} \left( \frac{ \text{Tr}(\Sigma (R + \gamma P_K (B^\top B+D^\top D)) - \frac{\tau}{2} (n + \log((2\pi)^n|\Sigma|)))}{1-\gamma} \right) \notag\\
    &  = (1-\gamma)^{-1}(R+ \gamma P_K(B^\top B+D^\top D))^\top  - \frac{\tau}{2} \Sigma^{-1}. \notag\\ 
\end{align*}
From \eqref{eq:state} we have
\begin{align*}
\mathbb{E}(x_{t+1}^2) = V_K \mathbb{E}x^2_t + \text{Tr}(\Sigma(B^\top B+D^\top D)).
\end{align*}
From \eqref{eq: P_k = ...} we have
\begin{align*}
    \nabla_K P_K &= 2RK + \gamma P_K \nabla V_K + \gamma \nabla P_KV_K\\
    & = (2RK + \gamma P_K\nabla V_K) \sum_{t=0}^{\infty} (\gamma V_K)^t,
\end{align*}
and 
$$
\nabla_Kq_{K,\Sigma} =  \sum_{t=0}^{\infty}\gamma^{t+1} \text{Tr} (\Sigma(B^\top B+D^\top D))\nabla_K P_K.
$$
Combining the above we have 
\begin{align*}
\nabla_K f({K,\Sigma}) & = (2RK + \gamma P_K \nabla V_K)\mathbb{E}x_0^2 + \gamma  \nabla_K P_K  V_K\mathbb{E}x_0^2) + \sum_{t = 0}^\infty \gamma^{t+1} \text{Tr}(\Sigma (B^\top B + D^\top D)) \nabla_K P_K\\ 
& = (2RK + \gamma P_K \nabla V_K)\mathbb{E}x_0^2  + \gamma (V_K \mathbb{E}x_0^2 +  \text{Tr}(\Sigma (B^\top B + D^\top D))\nabla P_K \\
&\quad +  \sum_{t = 1}^\infty \gamma^{t+1} \text{Tr}(\Sigma (B^\top B + D^\top D)) \nabla P_K \\
& = (2RK + \gamma P_K \nabla V_K)\mathbb{E} x_0^2  + \gamma \mathbb{E} x_1^2 \nabla_K P_K +  \sum_{t = 1}^\infty \gamma^{t+1} \text{Tr}(\Sigma (B^\top B + D^\top D)) \nabla_K P_K \\
& =  (2RK + \gamma P_K \nabla V_K)(\mathbb{E}x_0^2 + \gamma \mathbb{E}x_1^2)  \\
&\quad + \gamma (\mathbb{E} x_1^2 \nabla P_K + \sum_{t = 0}^\infty \gamma^{t+1} \text{Tr}(\Sigma (B^\top B + D^\top D)) \nabla P_K)\\
& = E_K \sum_{t=0}^{\infty }(\gamma^t \mathbb{E}x_t^2 )\\
& = E_KS_{K,\Sigma}.
\end{align*}
The proof is completed.
\end{proof}

\section{Proofs in Section \ref{section: model based}}

\subsection{Proof of Lemma \ref{lemma: gradient domination}}
\begin{proof} 

The proof is divided into the following steps

\textbf{Definition of advantage}
For any policy $K, K'$ that have the finite cost,  denote their trajectories as $x_t$ and $x'_t$ respectively. When $x_0 = x'_0$,we have the following,
\begin{align*}
& C_{K',\Sigma'}(x_0) - C_{K,\Sigma}(x_0) \\
& = \mathbb{E}_{\pi'} \Biggl\{ \sum_{t=0}^{\infty} \gamma^t \left[ (x'_ t)^2(Q+K'^\top RK') + \tau \log\pi'(u_t'|x'_t) \right] \Biggl\} -C_{K,\Sigma}(x_0)\\ 
				      & = \mathbb{E}_{\pi'} \Biggl\{ \sum_{t=0}^{\infty} \gamma^t \left[ (x'_ t)^2(Q+K'^\top RK') + \tau \log\pi'(u_t'|x'_t) -C_{K,\Sigma}(x'_t) + C_{K,\Sigma}(x'_t) \right] \Biggl\}-C_{K}(x_0)\\ 
				      & = \mathbb{E}_{\pi'} \Biggl\{ \sum_{t=0}^{\infty} \gamma^t \left[ (x'_ t)^2(Q+K'^\top RK') + \tau \log\pi'(u_t'|x'_t) -C_{K,\Sigma}(x'_t) \right] + \sum_{t=1}^{\infty} \gamma^t C_K(x'_t) \Biggl\}\\
                      & = \mathbb{E}_{\pi'} \Biggl\{ \sum_{t=0}^{\infty} \gamma^t \left[ (x'_ t)^2(Q+K'^\top RK') + \tau \log\pi'(u_t'|x'_t) + \gamma C_{K,\Sigma}(x'_{t+1}) -C_{K,\Sigma}(x'_t) \right] \Biggl\}\\
				      & \triangleq \mathbb{E}_{\pi'} \Biggl\{ \sum_{t=0}^{\infty}  \gamma^t A_{K,\Sigma
                      }(x'_t, u_t') \Biggl\},
\end{align*}
where $A_K(x, K')$ is called ``advantage", which can be viewed as the change in cost starting at state $x$ between one if choose $ u_t = -K'x$ only at current time and then $u_t = -Kx_{t}$ for all $t$ after the current time (i.e., $x^2(Q+K'^\top RK') + C_K(x_{t+1})$) and one if $u_t = -Kx_t$ for all $t$ (i.e., $C_K(x)$).
We now want to find $\mathbb{E}_{\pi'} \left[ A_{K,\Sigma} (x,u') \right]$.
\begin{align*}
  \mathbb{E}_{\pi'} \left[ A_K(x,K') \right] & = \mathbb{E}_{\pi'} \Biggl\{ \sum_{t=0}^{\infty} \gamma^t \left[ (x'_ t)^2(Q+K'^\top RK') + \tau \log\pi'(u_t'|x'_t)  \right] \Biggl\} \\ & +\gamma \mathbb{E} \biggl\{ P_K\left[ (A+w^xC)x+(B+w^uD)u  \right]^2 \biggl\} -P_K x^2-q_{K,\Sigma}\\
                        &  = (Q+K'^\top RK')x^2 + Tr(\Sigma'R) - \frac{\tau}{2}(n+log(2\pi)^ n|\Sigma'|) - P_kx^2 - (1-\gamma)q_{K,\Sigma}\\
                        &  + \gamma\left[ P_K(A^2+C^2+ K'^\top (B^\top B+D^\top D)K' - ABK')x^2 + Tr(\Sigma'P_K(B^\top B+D^\top D)) \right]\\
                        &  = (Q+K'^\top RK' + \gamma P_KV_K')x^2 - \frac{\tau}{2}(n+log(2\pi)^ n|\Sigma'|) - P_kx^2 - (1-\gamma)q_{K,\Sigma} \\
                        &  + Tr(\Sigma'(R+\gamma P_K(B^\top B+D^\top D))) \\
                        &  = (Q+K'^\top RK' + \gamma P_KV_K')x^2 -P_Kx^2 + (\gamma-1)(q_{K,\Sigma} - q_{K,\Sigma'}). 
\end{align*}

\textbf{Cost Difference}\\
\begin{align*}
\mathbb{E}_{\pi'} \left[A_K(x,K') \right] & = (Q+K'^\top RK' + \gamma P_KV_K')x^2 -P_Kx^2 + (\gamma-1)(q_{K,\Sigma} - q_{K,\Sigma'})\\
		& = x^2[Q + K'^\top RK' + \gamma P_K V_{K'} - (Q + K^\top RK + \gamma P_KV_K)] + (\gamma-1)(q_{K,\Sigma} - q_{K,\Sigma'}) \\
		& = x^2[K'^\top RK' - K^\top RK + \gamma P_K(V_{K'} -V_{K})] + (\gamma-1)(q_{K,\Sigma} - q_{K,\Sigma'})\\
		& = x^2[K'^\top RK' - K^\top RK - 2\gamma P_K AB(K'-K)\\
		&\quad + \gamma P_K (K'^\top (B^\top B+D^\top D)K' - K^\top (B^\top B+D^\top D)K) ] + (\gamma-1)(q_{K,\Sigma} - q_{K,\Sigma'})\\
		& = x^2[(K + K' -K)^\top R(K + K' -K) - K^\top RK - 2\gamma P_K AB(K'-K)\notag \\
		&\quad + \gamma P_K ((K + K' -K) ^\top (B^\top B+D^\top D)(K + K' -K) \\
		&\quad -  \gamma P_K K^\top (B^\top B+D^\top D)K) ] + (\gamma-1)(q_{K,\Sigma} - q_{K,\Sigma'})\\
		& = x^2[K' -K)^\top (R + \gamma P_K(B^\top B+D^\top D))(K' -K) \\
		& \quad + 2(K'-K)(R + \gamma P_K(B^\top B+D^\top D))K-\gamma P_K AB^\top  ] + (\gamma-1)(q_{K,\Sigma} - q_{K,\Sigma'})\\
		& = x^2[(K' -K)^\top (R + \gamma P_K(B^\top B+D^\top D))(K' -K)+ 2(K'-K)^\top E_K] \\
            & +(\gamma-1)(q_{K,\Sigma} - q_{K,\Sigma'}).
\end{align*}
as $\mathbb{E}_{\pi'} \left[A_K(x,K') \right]$ is in a quadratic form of $K' - K$, we have, 
\begin{align*}
\mathbb{E}_{\pi'} \left[A_K(x,K') \right] & \geq -x^2[E_K^\top (R + \gamma P_K(B^\top B+D^\top D))^{-1}E_K] + (\gamma-1)(q_{K,\Sigma} - q_{K,\Sigma'})
\end{align*}
with equality when $K' -K = -(R + \gamma P_K(B^\top B+D^\top D))^{-1}E_K$.\\

\textbf{Upper Bound}\\
Let $S_K = \sum_{t=0}^{\infty}\gamma^t \mathbb{E}[x_t^2]$ and remember that $\nabla_K f(K,\Sigma) = S_KE_K$. \\

\begin{align*}
&C_{K,\Sigma}(x_0) - C_{K^*,\Sigma^*}(x_0) \\
&  =  -\mathbb{E}_{\pi'} \left[ \sum_{t=0}^{\infty}  \gamma^t A_{K}(x^*_t, \pi^*) \right] \\
				       & = \sum_{t=0}^{\infty}  \gamma^t \mathbb {E} (x^*_t)^2[E_K^\top (R + \gamma P_K(B^\top B+D^\top D))^{-1}E_K] - (\gamma-1)(q_{K,\Sigma} - q_{K,\Sigma^*})\\
                       &  =  \sum_{t=0}^{\infty}  \gamma^t \mathbb {E} (x^*_t)^2[E_K^\top (R + \gamma P_K(B^\top B+D^\top D))^{-1}E_K] +(q_{K,\Sigma} - q_{K,\Sigma^*})\\
                          & = S_{K^*, \Sigma^*} E^\top _K(R+\gamma P_K(B^\top B + D^\top D))^{-1}E_K + (q_{K,\Sigma^*} - q_{K,\Sigma}) \\
				       & \leq \frac{S_{K^*, \Sigma^*}}{\sigma_{min}(R)} {E_K^\top E_K } +Tr(\nabla_{\Sigma} q_{K,\Sigma}^\top (\Sigma-\Sigma^*))\\
                       &  \leq \frac{S_{K^*,\Sigma^*}}{\mu^2 \sigma_{min}(R)} \nabla^\top _K f_{K,\Sigma}(x_0) \nabla_K f_{K,\Sigma}(x_0)  \\
                          & \quad + Tr[\nabla_{\Sigma} C_{K,\Sigma}(x_0)((R+\gamma P_K(B^\top B+D^\top D)))^{-1}((R+\gamma P_K(B^\top B+D^\top D))^\top -\frac{\tau}{2}\Sigma^{-1} )\Sigma]\\
                          &  \leq\frac{S_{K^*,\Sigma^*}}{\mu^2 \sigma_{min}(R)} \nabla^\top _K f_{K,\Sigma}(x_0) \nabla_K f_{K,\Sigma}(x_0)   \\
                          & \quad  + (1-\gamma) \text{Tr}[\nabla_{\Sigma} C_{K,\Sigma}(x_0)((R+\gamma P_K(B^\top B+D^\top D)))^{-1}\nabla_{\Sigma} C_{K,\Sigma}(x_0)] \\
                          & \leq\frac{S_{K^*,\Sigma^*}}{\mu^2 \sigma_{min}(R)} \nabla^\top _K f_{K,\Sigma}(x_0) \nabla_K f_{K,\Sigma}(x_0)   + \frac{(1-\gamma)\text{Tr}[(\nabla_{\Sigma} C_{K,\Sigma}(x_0))^2]}{\sigma_{min}(R)}.
    \end{align*}

As $q_{K,\Sigma}$ is a concave function w.r.t. $\Sigma$, so 
$q_{K,\Sigma^*} - q_{K,\Sigma}\leq \nabla_\Sigma q_{K,\Sigma}^\top  (\Sigma-\Sigma^*)$. $ \Sigma \preceq I$  

Now, taking the expectation w.r.t $x_0$ on both sides we have,
\begin{equation}
 f(K, \Sigma) - f(K^*, \Sigma)   \leq \frac{1}{ \mu \sigma_{min}(R)} \nabla_K f^\top (K, \Sigma) \nabla_K f(K,\Sigma) +\frac{(1-\gamma)\|\nabla_{\Sigma}f(K,\Sigma)\|^2}{\sigma_{min}(R)}.
\end{equation}
 
\textbf{Lower Bound} $C_{K^*}(x_0) \leq C_{K'}(x_0)$ for any $K' \in \mathbb{R}^n$,  we considering when 
$K' = K - (R + \gamma P_K(B^\top B+D^\top D))^{-1}E_K$
\begin{align*}
C_{K}(x_0) - C_{K^*}(x_0) & \geq C_{K}(x_0) - C_{K'}(x_0) \\
					& =-\mathbb{E}_{\pi'} \left[ \sum_{t=0}^{\infty}  \gamma^t A_{K}(x'_t, K') \right] \\
					& = S_K[E_K^\top (R + \gamma P_K(B^\top B+D^\top D))^{-1}E_K] +h_K(\Sigma)-h_K(\Sigma')\\
					& \geq \frac{\mathbb{E}[x_0^2]}{S_K^2\|R + \gamma P_K(B^\top B+D^\top D)\|}\nabla_K C_{K,\Sigma}^\top (x_0) \nabla_KC_{K,\Sigma}(x_0)
					\end{align*}.
taking the expectation w.r.t. $x_0$ on both sides we have
\begin{equation*}
 f(K,\Sigma)- f(K^*,\Sigma^*) \geq  \frac{\mathbb{E}[x_0^2]}{S_K^2  \|R + \gamma P_K(B^\top B+D^\top D)\|} \nabla f^\top (K,\Sigma) \nabla f(K,\Sigma). 
\end{equation*}
\end{proof}


\subsection{Proof of Lemma \ref{lemma: norm bounds}}
\begin{proof}
We have 
\begin{align}
    \| \nabla_Kf(K,\Sigma) \| & = \| E_K S_{K, \Sigma} \| \notag\\
                              & \leq \| E_K \| \frac{f(K, \Sigma) - \Omega}{Q} \notag\\
                              & \leq \frac{f(K, \Sigma) - \Omega}{Q} \sqrt{\lambda_1^{-1}(f(K, \Sigma)-f(K^*, \Sigma^*))},
\end{align}
where $\Omega = Tr(\Sigma R) - \frac{\tau}{2} \left(n+log(2 \pi)^n |\Sigma| \right)$ and the last line is by Lemma \ref{lemma: gradient domination}

\begin{align*}
    \| \nabla_{\Sigma}f(K, \Sigma) \| & = \| (1- \gamma)^{-1} \left( (R+\gamma P_K(B^\top B+D^\top D))^\top  + \frac{\tau}{2} \| \Sigma^{-1} \| \right) \| \notag\\
    & \leq (1-\gamma)^{-1} \left( \| R+ \gamma P_K(B^\top B+D^\top D) \| + \frac{\tau}{2} \| \Sigma^{-1} \|\right) \notag\\
    & \leq (1-\gamma)^{-1} \left( \| R+ \gamma P_K(B^\top B+D^\top D) \| + \frac{\tau}{2 \sigma_{min}(\Sigma)}  \right).
\end{align*}
Where the last line is because $\| \Sigma^{-1}\| \leq \sigma_{max}(\Sigma^{-1}) =\frac{1}{\sigma_{min}(\Sigma)}$
\end{proof}


\subsection{Proof of Lemma \ref{lemma: almost smooth}}
\begin{proof}

\begin{align*}
C_{K'}(x_0) - C_{K}(x_0)& =  \sum_{t=0}^{\infty}  \gamma^t \mathbb{E}[A_{K}(x'_t, K')]\\
				&= \displaystyle S_{K',\Sigma'}[(K' -K)^\top (R + \gamma P_K(B^\top B+D^\top D))(K' -K)+ 2(K'-K)E_K] + q_{K',\Sigma'}-q_{K,\Sigma}\notag
\end{align*}
We now need to prove the smoothness of $q_{K,\Sigma}$ with respect to $\Sigma$ by showing: 
$$q_{K',\Sigma'}-q_{K,\Sigma} + Tr\left(\nabla_{\Sigma}q_{K,\Sigma}^\top (\Sigma-\Sigma')\right)\leq \frac{m}{2}Tr((\Sigma^{-1}\Sigma'-I)^2)$$ 
Observe
\begin{align*}
    q_{K,\Sigma}-q_{K,\Sigma'} + Tr\left(\nabla_{\Sigma}q_{K,\Sigma}^\top (\Sigma'-\Sigma)\right) = \frac{\tau}{2(1-\gamma)} \left[ log(\Sigma^{-1}\Sigma') -Tr(\Sigma^{-1}\Sigma'-I)\right]
\end{align*}
But since $\Sigma \text{ and }\Sigma'$ are positive definite so is $\Sigma^{-1}\Sigma'$. Then $\sigma_{min}(\Sigma^{-1}\Sigma')\geq\sigma_{min}(\Sigma^{-1})\sigma_{min}(\Sigma') \geq a >0$. In addition, $a\leq \lambda_1 \leq ...\leq \lambda_n$ where $\lambda_i's$ are the eigenvalues of $\Sigma^{-1}\Sigma'$. Note that $log(\Sigma^{-1}\Sigma') -Tr(\Sigma^{-1}\Sigma'-I) = \sum_{i=1}^n log(\lambda_i) + \lambda_i -1 \leq m\sum_{i=1}^n(\lambda_i-1)^2$ so let $m = \frac{log(a)+a-1}{(a-1)^2}$. Taking the expectation of $x_0$ on both sides completes the proof.
\end{proof}


\subsection{Proof of Lemma \ref{lemma: bound of sigma}}

\begin{proof}
    We will first show that 
    $$ aI \preceq \Sigma - \frac{\eta}{1- \gamma}(R-\frac{\tau}{2}\Sigma^{-1}+\gamma P_K(B^\top B+D^\top D)) \prec I$$
    Let $h(y) = y+ \frac{\tau}{2(1-\gamma)y}$ which is monotonic increasing for $y \in \left[ \sqrt{\frac{\eta \tau}{2(1-\gamma)}}, \infty \right)$ as $\sqrt{\frac{\eta \tau}{2(1-\gamma)}} \leq a \leq \frac{\sigma_{min}(R)}{\| R+ \gamma P_K(B^\top B+D^\top D) \|} < 1$. Now observe: 
    $$\Sigma + \frac{\eta\tau}{2(1-\gamma)} \Sigma' - \frac{\eta}{1-\gamma}(R+\gamma P_K(B^\top B+D^\top D)) \succeq \left(a+ \frac{\eta \tau}{2(1-\gamma)a} \right) - \frac{\eta}{(1-\gamma)}(R+\gamma P_K(B^\top B+D^\top D))$$ 
    $$ \succeq \left(a + \frac{\eta}{1-\gamma} \| R+\gamma P_K(B^\top B+D^\top D) \| \right)I - \frac{\eta}{1-\gamma} (R+\gamma P_K(B^\top B+D^\top D)) \succeq aI $$ 
    Now,
    $$\Sigma + \frac{\eta \tau}{2(1-\gamma)} \Sigma^{-1} - \frac{\eta}{1- \gamma}(R+\gamma P_K(B^\top B+D^\top D)) \preceq \left( 1 + \frac{\eta \tau}{2(1- \gamma)} \right)I - \frac{\eta}{1- \gamma}(R+ \gamma P_K(B^\top B+D^\top D))$$ 
    $$ \preceq \left( 1+ \frac{\eta}{1-\gamma} \sigma_{min}(R) \right)I - \frac{\eta}{1- \gamma} (R+ \gamma P_K(B^\top B+D^\top D)) \preceq I $$ 
    And so using these facts we can write, 
    $$aI \preceq \Sigma - \frac{\eta}{1- \gamma}(R-\frac{\tau}{2}\Sigma^{-1}+\gamma P_K(B^\top B+D^\top D)) \prec I$$ \\
    Next I will show $aI \preceq \Sigma' \preceq I$. Observe that: \\
    $$aI - \Sigma \preceq - \frac{\eta}{1- \gamma}(R+ \gamma P_K(B^\top B+D^\top D) - \frac{\tau}{2} \Sigma^{-1}) \preceq I - \Sigma$$ 
    Now multiply both sides by $\Sigma$ and add $\Sigma$ which yields: 
    $$a \Sigma^2 - \Sigma^3 + \Sigma \preceq \Sigma - \frac{\eta}{1- \gamma}\Sigma \left(R+ \gamma P_K(B^\top B+D^\top D) - \frac{\tau}{2}\Sigma^{-1} \right) \Sigma \preceq \Sigma^2 - \Sigma^3 + \Sigma$$
$a<\sigma_{min}(\Sigma)$ so we have $aI -\Sigma\preceq 0$.
$aI - \Sigma \preceq  (aI-\Sigma)\Sigma^2 $ as $\Sigma \preceq I$.
$a\Sigma^2 -\Sigma^3 +\Sigma \succeq aI$.

As $I - \Sigma \succeq 0$, we have $I-\Sigma \succeq (I - \Sigma)\Sigma^2$, so $\Sigma^2-\Sigma^3 +\Sigma\preceq I$
The proof is completed
\end{proof}


\subsection{Proof of Lemma \ref{lemma: contraction of RPG}}

\begin{proof}
\begin{align*}
    f(K',\Sigma') - f(K,\Sigma)
    &= S_{K'} [(K' -K)^\top (R + \gamma P_K(B^\top B+D^\top D))(K' -K)+ 2(K'-K)^\top E_K] \notag\\
    & \quad  +
    q_{K',\Sigma'} - q_{K,\Sigma}
\end{align*}

Using  RPG  and  $\eta_1 \leq \frac{1}{\|R + \gamma P_K(B^\top B+D^\top D)\|}$ we have 
\begin{align*}
& S_{K'} [(K' -K)^\top (R + \gamma P_K(B^\top B+D^\top D))(K' -K)+ 2(K'-K)^\top E_K] \\
& \quad \leq S_{K'}[\eta_1^2E_K^\top (R + \gamma P_K(B^\top B+D^\top D)) E_K - 2\eta_1E_K^\top E_K]\\
& \quad \leq -\eta_1  S_{K'} E_K^\top E_K\\
& \quad \leq -\eta_1 \mu \frac{\sigma_{min}{(R)}}{S_{K^*, \Sigma^*}}E^\top _K(R+\gamma P_K(B^\top B + D^\top D))^{-1}E_K\\
\end{align*}
From lemma \ref{lemma: almost smooth}, we have
\begin{align*}
q_{K,\Sigma'} - q_{K,\Sigma} & \leq \frac{\text{Tr}\left( ((R+ \gamma P_K(B^\top B+D^\top D)) - \frac{\tau}{2} \Sigma^{-1})(\Sigma'-\Sigma)\right)}{(1- \gamma)}+ \frac{\tau m}{2(1-\gamma)} Tr((\Sigma^{-1}\Sigma'- I)^2)\\
& =\frac{-\eta_2}{{(1- \gamma)}} {\text{Tr}[\left( ((R+ \gamma P_K(B^\top B+D^\top D)) - \frac{\tau}{2} \Sigma^{-1})\Sigma\right)^2]}\\
& \quad + \frac{\eta_2^2\tau m}{2(1-\gamma)^3} \text{Tr}( ( (R+ \gamma P_K(B^\top B+D^\top D))\Sigma - \frac{\tau}{2}I )^2)\\
& \leq -\frac{\eta_2}{2(1-\gamma)^2}\text{Tr}[(R + \gamma P_K (B^\top B+D^\top D))^2],
\end{align*}
$$
\eta_2 \leq \frac{2(1-\gamma)a^2}{\tau}\leq \frac{2(1-\gamma)}{\tau}\left(\frac{\tau}{2\|R+\gamma P_K (B^\top B+D^\top D)\|}\right)\leq \frac{2(1-\gamma)}{\tau} \leq \frac{(1-\gamma)}{\tau m},
$$

From lemma \ref{lemma: gradient domination}, we have
\begin{align*}
q_{K,\Sigma^*} - q_{K,\Sigma}& \leq Tr[\nabla_{\Sigma} C_{K,\Sigma}(x_0)((R+\gamma P_K(B^\top B+D^\top D)))^{-1}((R+\gamma P_K(B^\top B+D^\top D))-\frac{\tau}{2}I
)]\\
&\leq \frac{1}{(1-\gamma) \sigma_{min}(R)}Tr[((R+\gamma P_K(B^\top B+D^\top D))-\frac{\tau}{2}I)^2\Sigma^{-1}]\\
&\leq \frac{1}{(1-\gamma) a\sigma_{min}(R)}Tr[((R+\gamma P_K(B^\top B+D^\top D))-\frac{\tau}{2}I)^2].
\end{align*}
Combining the above we have 
\begin{align*}
q_{K,\Sigma'} - q_{K, \Sigma} \leq \frac{\eta_2a\sigma_{min}(R)}{2(1-\gamma)}(q_{K,\Sigma} - q_{K,\Sigma^*}).
\end{align*}
Finally, with $\phi = min\{ \eta_1 \mu \frac{\sigma_{min}{(R)}}{S_{K^*,\Sigma^*}}, \frac{\eta_2 a\sigma_{min}(R)}{2(1-\gamma)}\}$, we have
\begin{align*}
& f(K', \Sigma') - f(K, \Sigma)\\
 & \quad \leq -\eta_1 \mu \frac{\sigma_{min}{(R)}}{S_{K^*,\Sigma^*}}E^\top _K(R+\gamma P_K(B^\top B + D^\top D))^{-1}E_K + \frac{\eta_2 \sigma_{min} (R)}{2(1-\gamma)} (q_{K,\Sigma} - q_{K, \Sigma^*})\\
 & \quad \leq -\phi \left({S_{K^*,\Sigma^*}}E^\top _K(R+\gamma P_K(B^\top B + D^\top D))^{-1}E_K + (q_{K,\Sigma^*} - q_{K, \Sigma})\right)\\
  & \quad \leq -\phi (f(K,\Sigma) - f(K^*, \Sigma^*))
\end{align*}
and
$$
f(K', \Sigma') - f(K^*, \Sigma^*) \leq (1-\phi) (f(K,\Sigma) - f(K^*, \Sigma^*)).
$$
\end{proof}


\subsection{Proof of Lemma \ref{lemma: lower bodun of f}}
\begin{proof}
\begin{align*}
  q_{K,\Sigma}  &=  \frac{ \text{Tr}(\Sigma (R + \gamma P_K (B^\top B+D^\top D)) - \frac{\tau}{2} (n + \log((2\pi)^n|\Sigma|)))}{1-\gamma}\\
  &\geq \frac{1}{1-\gamma}[\sigma_{min}(R)\text{Tr}(\Sigma) - \frac{\tau}{2}(k + klog(2\pi) + \log|\Sigma|)]\\
  & \geq \frac{1}{1-\gamma} \left[\frac{\tau k}{2} - \frac{\tau}{2}(k + log(2\pi)) - \frac{\tau K}{2}log \left(\frac{\tau}{2\sigma_{min}(R)} \right) \right].
\end{align*}
As $\frac{\tau k}{2} - \frac{\tau}{2}(k + log(2\pi)) - \frac{\tau k}{2}log(\frac{\tau}{2\sigma_{min}(R)})$ is a convex function w.r.t. $\Sigma$ with minimizer $\frac{\tau}{2\sigma_{min}(R)}I$, so we have,
$$
q_{K,\Sigma} \geq \frac{\tau k}{2(1-\gamma)}log \left(\frac{\sigma_{min}(R)}{\pi \tau} \right).
$$

\end{proof}


\section{Proofs in Section \ref{section: model free}}

\subsection{Proof of Lemma \ref{lemma:upper bound of S_K}}
\begin{proof}
\begin{align*}
    f(K,\Sigma) &= \mathbb{E}_\pi\left[\sum_{t=0}^\infty \gamma^t ( Q x_t^2 + u_t^{T} R u_t)\right] - \frac{\frac{\tau}{2}( n + log(2\pi)^n|\Sigma| )}{1-\gamma} \\
                & = (Q+K^\top RK)S_K + \frac{Tr(\Sigma R) - \frac{\tau}{2} \left( n + log(2\pi)^n|\Sigma| \right)}{1-\gamma}. \\
\end{align*}
Then we have 
\begin{align*}
    S_K & = \frac{f(K,\Sigma) - (1-\gamma)^{-1}\left[ Tr(\Sigma R) - \frac{\tau}{2} \left( n + log(2\pi)^n|\Sigma| \right)\right]}{Q+K^\top RK} \\
        & \leq \frac{f(K, \Sigma)-(1-\gamma)^{-1}\left[ Tr(\Sigma R) - \frac{\tau}{2} \left( n + log(2\pi)^n|\Sigma| \right)\right]}{Q}
\end{align*}
and
\begin{align*}
    \mathbb{E}[x^2_{t+1}] & = V_K \mathbb{E}[x_t^2] + Tr(\Sigma(B^\top B+D^\top D)) \\
                          & =V_K (V_K \mathbb{E}[x_{t-1}^2] + Tr(\Sigma(B^\top B+D^\top D))) + Tr(\Sigma(B^\top B+D^\top D))) \\
                          & = V_K^2 \mathbb{E}[x_{t-1}^2] + (V_K+1) Tr(\Sigma(B^\top B+D^\top D)) \\
                          & = V_K^{t+1} \mu + Tr(\Sigma(B^\top B+D^\top D)) \sum_{i=0}^{t} V_K^i.
\end{align*}
Observe that 
\begin{align*}
    S_K & = \sum_{t=0}^{\infty} \gamma^t \mathbb{E}[x_t^2] \\
        & = \sum_{t=0}^{\infty} \gamma^t \left( V_K^t \mathbb{E}[x_0^2] + Tr(\Sigma(B^\top B+D^\top D)) \frac{1-V_K^t}{1-V_K} \right) \\
        & = \mu \sum_{t=0}^{\infty} (\gamma V_K)^t + \frac{Tr(\Sigma(B^\top B+D^\top D))}{1-V_K} \sum_{t=0}^{\infty} \gamma^t(1-V_K^t) \\
        & = \mu\sum_{t=0}^{\infty} (\gamma V_K)^t + \frac{Tr(\Sigma(B^\top B+D^\top D))}{1-V_K} \left[ \frac{1}{1-\gamma} -\frac{1}{1-\gamma V_K} \right].
\end{align*}
\end{proof}
\subsection{Proof of Lemma \ref{lemma: approximate fk}}
\begin{proof}
Note that
\begin{align}
 f(K, \Sigma) & = \mathbb{E}[P_Kx_0^2 + q_{K,\Sigma}] \notag\\
              & = (Q+K^\top RK)S_K + (1- \gamma)^{-1} \left[ Tr(\Sigma R) - \frac{\tau}{2} (n+log(2 \pi)^n|\Sigma|) \right] \notag \\ 
S_K & = \mathbb{E}[x_0^2] \sum_{t=0}^{\infty} (\gamma V_K)^t + \frac{Tr(\Sigma(B^\top B+D^\top D))}{1-V_K} \sum_{t=0}^{\infty} \gamma^t(1-V_K^t),
\end{align}
So
\begin{align}
 f^{(l)}(K, \Sigma) & = (Q+K^\top RK)S^{(l)}_K + \frac{1-\gamma^l}{1- \gamma} \left[ Tr(\Sigma R) - \frac{\tau}{2} (n+log(2 \pi)^n|\Sigma|) \right] \notag \\ 
S^{(l)}_{K,\Sigma} & = \mu \sum_{t=0}^{l-1}(\gamma V_K)^t + \frac{Tr(\Sigma(B^\top B+D^\top D))}{1-V_K} \sum_{t=0}^{l-1} \gamma^t(1- V_K^t).
\end{align}
Next note that
\begin{align*}
    S_K - S_K^{(l)} & = \mu \sum_{t=l}^{\infty}(\gamma V_K)^t + \frac{Tr(\Sigma(B^\top B+D^\top D))}{1-V_K} \sum_{t=l}^{\infty} \gamma^t(1- V_K^t) \\
    & = (\mu - \frac{Tr(\Sigma(B^\top B+D^\top D))}{1-V_K})\frac{(\gamma V_K)^l }{1-\gamma V_K} + \frac{Tr(\Sigma(B^\top B+D^\top D))}{1-V_K}\frac{\gamma^l }{1-\gamma}\\
    &  \leq  (\mu - \frac{Tr(\Sigma(B^\top B+D^\top D))}{1-V_K})\frac{\gamma^l }{1-\gamma V_K} + \frac{Tr(\Sigma(B^\top B+D^\top D))}{1-V_K}\frac{\gamma^l }{1-\gamma}\\
    & = \gamma^l S_{K,\Sigma}
\end{align*}
and 
\begin{align*}
f(K, \Sigma) - f^{(l)}(K, \Sigma) & = (Q + K^\top RK) \sum_{t=l}^{\infty}\mathbb{E} \gamma^t x_t^2 +(1-\gamma)^{-1} \left[ Tr(\Sigma R) - \frac{\tau}{2}(n+ log(2\pi)^n|\Sigma|)\right] \gamma^l  \notag\\
		      & = (Q + K^\top RK) (S_K- S_K^{(l)}) +(1-\gamma)^{-1} \left[ Tr(\Sigma R) - \frac{\tau}{2}(n+ log(2\pi)^n|\Sigma|)\right] \gamma^l \notag\\
		      & \leq (Q + K^\top RK) \gamma^l S_{K,\Sigma} +(1-\gamma)^{-1} \left[ Tr(\Sigma R) - \frac{\tau}{2}(n+ log(2\pi)^n|\Sigma|)\right] \gamma^l \\
              & = \gamma^l\left[ (Q + K^\top RK)S_{K,\Sigma} + \frac{ Tr(\Sigma R) - \frac{\tau}{2}(n+ log(2\pi)^n|\Sigma|)}{1-\gamma }\right].
\end{align*}
taking $l$ in the above inequality completes the proof.
\end{proof}

\subsection{Two useful Lemmas}

\begin{lemma}\label{lemma: bound transform}
If  $ \frac{|\gamma V_K - \gamma V_{K'}|}{1-\gamma V_K}  \leq \frac{1}{2} $, then
\begin{align*}
|(1-\gamma V_K)^{-1} - (1-\gamma V_{K'})^{-1}| & \leq  2(1-\gamma V_K)^{-2} |\gamma V_{K'} - \gamma V_K|.
\end{align*}
\end{lemma}

\begin{proof}
We have 
$$
(1 - (1-\gamma V_K)^{-1}(\gamma V_{K'}-\gamma V_K))^{-1} \leq (1 - (1-\gamma V_K)^{-1}|\gamma V_{K'}-\gamma V_K|)^{-1} \leq 2
$$
as $ (1-\gamma V_K)^{-1} |\gamma V_K - \gamma V_{K'}| \leq \frac{1}{2}$.

\begin{align}
(1-\gamma V_K)^{-1} - (1-\gamma V_{K'})^{-1}& = (1-\gamma V_K)^{-1} - [(1-\gamma V_{K}) - (\gamma V_{K'}-\gamma V_{K})]^{-1}\notag\\
								      & = (1-\gamma V_K)^{-1}[1 - (1 - (1-\gamma V_K)^{-1}(\gamma V_{K'}-\gamma V_K))^{-1})]\notag
\end{align}
\begin{align}
 |1 - (1 - (1-\gamma V_K)^{-1}(\gamma V_{K'}-\gamma V_K))^{-1}| & = (1-\gamma V_K)^{-1}|(\gamma V_{K'}-\gamma V_K)  (1 - (1-\gamma V_K)^{-1}(\gamma V_{K'}-\gamma V_K))^{-1}| \notag\\
 & \leq 2 (1-\gamma V_K)^{-1}|\gamma V_{K'}-\gamma V_K| \notag.
\end{align}
So we have $ (1-\gamma V_K)^{-1} - (1-\gamma V_{K'})^{-1}  \leq  2(1-\gamma V_K)^{-2} |\gamma V_{K'} - \gamma V_K|$
\end{proof}

\begin{lemma} \label{lemma: per pre}If
$
\|K-K'\|\leq h_\Sigma
$,then
$$
 |\frac{1}{1-\gamma V_{K'}}-\frac{1}{1-\gamma V_{K}}|\leq g_\Sigma \|K'-K\|,
 $$
where $h_{\Sigma}$ is defined in Lemma \ref{lemma: sk perturbation}
\end{lemma}
\begin{proof}
\begin{align}
|\gamma V_K - \gamma V_{K'}| & = (A^2+C^2 + K^\top (B^\top B+D^\top D)K -2ABK) \notag\\
					         & \quad - (A^2+C^2 + K'^\top (B^\top B+D^\top D)K' -2ABK')\notag\\
					         & = K^\top (B^\top B+D^\top D)K -2AB(K-K') \notag\\
					         &\quad -  (K + K' - K)^\top (B^\top B+D^\top D)(K+K'-K)\notag\\
					         & =  -(K - K')^\top (B^\top B+D^\top D)(K-K') -2AB(K-K')\notag \\
					         &\quad +2K^\top (B^\top B+D^\top D)(K-K')\notag \\
					         & \leq 2\|K^\top (B^\top B + D^\top D) - AB\|\|K-K'\|\notag\\
					         & \quad  + \|(B^\top B+D^\top D)\|\|K-K'\|^2\notag\\
					         & = \|K-K'\|(2 \|K^\top (B^\top B + D^\top D) - AB\| +\|B^\top B+D^\top D\|\|K-K'\| )\notag\\
					         & = \|B^\top B+D^\top D\|\left(\|K-K'\| + \frac{ \|K^\top (B^\top B + D^\top D) - AB\|}{ \|B^\top B+D^\top D\|}\right)^2 \notag\\
					         & \quad- \frac{\|K^\top (B^\top B + D^\top D) - AB\|^2}{ \|B^\top B+D^\top D\|} \notag.
\end{align}
If we set 
\begin{align}
\|K-K'\| &\leq \sqrt{\frac{1}{2}\frac{(1-\gamma V_K)^2}{ \|B^\top B+D^\top D\| } + \left(\frac{ \|K^\top (B^\top B + D^\top D) - AB\|}{ \|B^\top B+D^\top D\|}\right)^2} -\frac{ \|K^\top (B^\top B + D^\top D) - AB\|}{ \|B^\top B+D^\top D\|} \notag\\
	& = \frac{\frac{1}{2}\frac{(1-\gamma V_K)^2}{\|B^\top B+D^\top D\| } }{\sqrt{\frac{1}{2}\frac{(1-\gamma V_K)^2}{ \|B^\top B+D^\top D\| } + \left(\frac{ \|K^\top (B^\top B + D^\top D) - AB\|}{ \|B^\top B+D^\top D\|}\right)^2}  + \frac{ \|K^\top (B^\top B + D^\top D) - AB\|}{ \|B^\top B+D^\top D\|} }\\
	& = \frac{\frac{1}{2}{(1-\gamma V_K)^2} }{\sqrt{\frac{1}{2}{(1-\gamma V_K)^2\|B^\top B+D^\top D\|} + \|K^\top (B^\top B + D^\top D) - AB\|^2} + \|K^\top (B^\top B + D^\top D) - AB\|}\notag\\
	& =  \frac{1}{2}\frac{(1-\gamma V_K)^2 }{\sqrt{\frac{1}{2}(1-\gamma V_K)^2 \|B^\top B+D^\top D\|+ \|K^\top (B^\top B + D^\top D) - AB\|^2} + \|K^\top (B^\top B + D^\top D) - AB\|}\notag\\
    & := h_\Sigma,
\end{align}
then we have
\begin{align}
\frac{1}{1-\gamma V_K} |\gamma V_K - \gamma V_{K'}| &\leq \frac{1}{1-\gamma V_K}  \frac{(1 - \gamma V_K)^2}{2}\notag \\
											  & \leq \frac{1-\gamma V_K}{2}\notag\\
											  & \leq \frac{1}{2},
\end{align}
which satisfies the condition in lemma \ref{lemma: bound transform}. Apply lemma \ref{lemma: bound transform} we have,
\begin{align*}
|\frac{1}{1-\gamma V_K} - \frac{1}{1-\gamma V_{K'}}| & \leq 2 \left(\frac{1}{1-\gamma V_K } \right)^2 |\gamma V_K - \gamma V_{K'}|\\
& \leq 2\left(\frac{1}{1-\gamma V_K } \right)^2 (2 \|K^\top (B^\top B + D^\top D) - AB\| +\|B^\top B+D^\top D\|\|K-K'\|)  \|K-K'\| \\
& \leq  2\left(\frac{1}{1-\gamma V_K } \right)^2 (2 \|K^\top (B^\top B + D^\top D) - AB\| +\|B^\top B+D^\top D\| h_\Sigma)  \|K-K'\|\\ 
& := g_\Sigma \|K-K'\|.
\end{align*}
\end{proof}
\subsection{Proof of Lemma \ref{lemma: sk perturbation}}
\begin{proof}
We have $\|\Sigma'\|  \leq 2\|\Sigma\|$
as $
\|\Sigma'\| - \| \Sigma\|\leq\|\Sigma' - \Sigma\| \leq \|\Sigma\|.
$
\begin{align*}
    S_{K,\Sigma} & = \mu \sum_{t=0}^{\infty} (\gamma V_K)^t + \frac{Tr(\Sigma(B^\top B+D^\top D))}{1-V_K} \sum_{t=0}^{\infty} \gamma^t(1-V_K^t) \notag\\
        & = \frac{\mu}{1- \gamma V_K}  + \frac{Tr(\Sigma(B^\top B+D^\top D))}{1-V_K} \left[ \frac{1}{1-\gamma} - \frac{1}{1- \gamma V_K} \right] \notag\\
        & = \frac{\mu}{1- \gamma V_K}  + \frac{Tr(\Sigma(B^\top B+D^\top D))}{1-V_K} \left[  \frac{(1-\gamma V_K)-(1-\gamma)}{(1-\gamma)(1- \gamma V_K)} \right] \notag\\
        & = \frac{\mu}{1- \gamma V_K}  + \frac{Tr(\Sigma(B^\top B+D^\top D))}{1-V_K} \left[  \frac{\gamma (1-V_K)}{(1-\gamma)(1- \gamma V_K)} \right] \notag\\
        & = \frac{(1- \gamma)\mu+ \gamma Tr(\Sigma(B^\top B+D^\top D))}{(1-\gamma)(1- \gamma V_K)}   \notag\\
        & := \Delta_\Sigma (1- \gamma V_K)^{-1}.
\end{align*}
\begin{align*}
 |S_{K',\Sigma'}-S_{K,\Sigma}|& \leq  |S_{K',\Sigma'}-S_{K, \Sigma'}| +| S_{K,\Sigma'} - S_{K,\Sigma}|\\
 				               & \leq \Delta_{\Sigma'}|\frac{1}{1-\gamma V_{K'}} - \frac{1}{1-\gamma V_{K}}| +  \frac{\gamma Tr((\Sigma' - \Sigma)(B^\top B+D^\top D))}{(1-\gamma)(1- \gamma V_K)}\\
				               & \leq \Delta_{\Sigma'} g_{\Sigma} \|K'-K\| +  \frac{\gamma Tr((B^\top B+D^\top D))}{(1-\gamma)(1- \gamma V_K)} \|\Sigma' - \Sigma \|\\
				               & \leq 2\Delta_{\Sigma} g_{\Sigma} \|K'-K\| +  \frac{\gamma Tr((B^\top B+D^\top D))}{(1-\gamma)(1- \gamma V_K)} \|\Sigma' - \Sigma \|\\
				               &: = h_K\|K'-K\| + h_2 \|\Sigma - \Sigma'\|.
\end{align*}

\end{proof}
\subsection{Proof of Lemma \ref{lemma: pk perturbation}}
\begin{proof}
From $\|K'\| - \|K\| \leq\|K-K'\| \leq |K|$ we have $K'\leq 2\|K\|$ and $K'RK' \leq \|R\|\|K'\|^2 \leq 4\|R\|\|K\|^2$
  \begin{align*}
|P_K  - P_{K'}| &= \bigg|\frac{Q + K^\top RK}{1-\gamma V_K} - \frac{Q + K'^\top RK'}{1-\gamma V_{K'}} \bigg| \\
& \leq \bigg|\frac{Q + K^\top RK}{1-\gamma V_K} - \frac{Q + K'^\top RK'}{1-\gamma V_{K}}\bigg| + \bigg|\frac{Q + K'^\top RK'}{1-\gamma V_K} -\frac{Q + K'^\top RK'}{1-\gamma V_{K'}}\bigg|\\
 &=\frac{|{K'}^\top RK' - K^\top RK|}{1-\gamma V_K} +(Q+{K'}^\top RK')\bigg| \frac{1}{1-\gamma V_K} - \frac{1}{1-\gamma V _{K'}}\bigg|\\
 & \leq \frac{|{K'}^\top RK' - K^\top RK|}{1-\gamma V_K}  + (Q+4\|R\|\|K\|^2) g_\Sigma \|K-K'\|\\
 & = \frac{ |({K'} - K)^\top R(K'-K) - 2K^\top R(K'-K)|}{1-\gamma V_K} + (Q+4\|R\|\|K\|^2) g_\Sigma \|K-K'\|\\
 & \leq \frac{3\|K\|\|R\|}{1-\gamma V_K}\|K'-K\| + (Q+4\|R\|\|K\|^2) g_\Sigma \|K-K'\|\\
 & := h_5\|K'-K\|.
\end{align*}  
\end{proof}
\subsection{Proof of Lemma \ref{lemma: gradient perturbation}}
\begin{proof}
We first consider $\nabla_K f(K,\Sigma)$ perturbation 
\begin{align*}
\nabla_K f(K', \Sigma') - \nabla_K f(K, \Sigma)  & = E_{K'} S_{K', \Sigma'} -  E_{K} S_{K, \Sigma}  \\
				    & = (E_{K'} - E_K)S_{K', \Sigma'} - (S_{K', \Sigma'} - S_{K, \Sigma})E_K \\
                    & \leq (E_{K'} - E_K)S_{K', \Sigma'} + (S_{K', \Sigma'} - S_{K, \Sigma})E_K.
\end{align*}
By triangle equality we have
\begin{align*}
\|\nabla_K f(K', \Sigma') - \nabla_K f(K, \Sigma)\|  & \leq \|E_{K'} - E_K\||S_{K', \Sigma'}| + |S_{K', \Sigma'} - S_{K, \Sigma}| E_K.
\end{align*}
By definition of $E_K$, we have
\begin{align}
\frac{1}{2}\| E_{K'} - E_K\| & = \|R\|\|K' - K\| + \gamma P_{K'}[(B^\top  B+D^\top  D)K'] - \gamma P_{K}[(B^\top  B+D^\top  D)K]  \notag\\
				        & \quad +   \gamma |P_{K'} - P_K|A\|B\|\notag\\
				        & \leq \|R\|\|K' - K\| +\gamma A \cdot h_5 \|B\| \|K'-K\|+ \gamma P_{K'}(B^\top  B+D^\top  D)K' \notag\\
				        & \quad -\gamma P_{K}(B^\top  B+D^\top  D)(K-K'+K') \notag\\
				        & \leq \|R\|\|K' - K\| +\gamma A \cdot h_5 \|B\| \|K'-K\|+ \gamma P_{K}(B^\top  B+D^\top  D) \|K'-K\|\notag\\
				        & \quad +\gamma | P_{K'} -P_{K}|\|B^\top  B+D^\top  D\|\|K'\| \notag\\
				        & \leq \|R\|\|K' - K\| +\gamma A \cdot h_5 \|B\| \|K'-K\|+ \gamma P_{K}(B^\top  B+D^\top  D) \|K'-K\|\notag\\
				        & \quad + 2\gamma \cdot h_5\|B^\top  B+D^\top  D\|\|K\|\|K'-K\| \notag\\
				        & = \left[\|R\| + \gamma A \cdot h_5 \|B\| +\gamma P_{K}(B^\top  B+D^\top  D) + + 2\gamma \cdot h_5\|B^\top  B+D^\top  D\|\|K\| \right] \|K'-K\|.\notag\\
\| E_{K'} - E_K\|           &\leq 2\left[\|R\| + \gamma A \cdot h_5 \|B\| +\gamma P_{K}(B^\top  B+D^\top  D) + + 2\gamma \cdot h_5\|B^\top  B+D^\top  D\|\|K\| \right] \|K'-K\|\notag\\
                            &:= h_E \|K'-K\| \notag.
\end{align}				        
Using $S_{K', \Sigma'} - S_{K,\Sigma} \leq |S_{K', \Sigma'} - S_{K,\Sigma}|$ and lemma \ref{lemma: pk perturbation} we have 
\begin{align*}
    S_{K',\Sigma'} &\leq  S_{K,\Sigma} + h_K\|K'-K\| + h_2 \| \Sigma' - \Sigma\| \\
    & \leq  S_{K,\Sigma} + h_Kh_{\Sigma} + h_2 \|\Sigma\|,
\end{align*}
and 
\begin{align*}
\|E_{K'} - E_K\|S_{K', \Sigma'}\leq (S_{K,\Sigma} + h_Kh_{\Sigma} + h_2 \|\Sigma\|) h_E \|K'-K\|,
\end{align*}
if $\|K'-K\| \leq h_{\Sigma}$ and $\|\Sigma - \Sigma'\| \leq \|\Sigma\|$.

From Lemma \ref{lemma: gradient domination}, we have 
$$
\|E_K\|\leq \sqrt{\lambda_1^{-1}|f(K, \Sigma) - f(K^*, \Sigma^*)| } 
$$
so for the second term we have:
\begin{align*}
|S_{K', \Sigma'} - S_{K, \Sigma}|  \| E_K  \| \leq (h_K \| K' - K \| + h_2 \| \Sigma' - \Sigma\|) \sqrt{\lambda_1^{-1}|f(K, \Sigma) - f(K^*, \Sigma^*)| },
\end{align*}
by lemma \ref{lemma: sk perturbation}. Then we have
\begin{align*}
    \| \nabla_K f(K', \Sigma') - \nabla_K f(K, \Sigma) \| & \leq \left( h_K  \sqrt{\lambda_1^{-1}|f(K, \Sigma) - f(K^*, \Sigma^*)| } + (S_{K,\Sigma} + h_Kh_{\Sigma} + h_2 \|\Sigma\|) h_E  \right) \| K'-K \| \notag\\
    & + \left( h_2  \sqrt{\lambda_1^{-1}|f(K, \Sigma) - f(K^*, \Sigma^*)| } \right) \| \Sigma' - \Sigma \|.
\end{align*}

Now consider the perturbation of $\nabla_\Sigma f(K, \Sigma)$ \\
Apply Theorem 3.5 in \cite{Fazel2018Global}  we have 
\begin{align*}
  \| \Sigma'^{-1} - \Sigma^{-1} \| &  \leq \frac{2}{\sigma_{min}(\Sigma)} \|\Sigma - \Sigma'\|
\end{align*}
if $\displaystyle \|\Sigma - \Sigma'\| \leq \frac{\sigma_{min}(\Sigma)}{2}$.
\begin{align*}
    \nabla_{\Sigma} f(K',\Sigma') - \nabla_{\Sigma}f(K, \Sigma) & = \frac{\left( (R+ \gamma P_{K'}(B^\top B+D^\top D))^\top  - \frac{\tau}{2} \Sigma'^{-1}\right) - \left( (R+ \gamma P_K(B^\top B+D^\top D)) - \frac{\tau}{2} \Sigma^{-1}\right) }{1-\gamma} \notag \\
    & = \frac{\gamma \left( P_{K'} - P_K \right) (B^\top B +D^\top D) - \frac{\tau}{2}(\Sigma'^{-1} - \Sigma^{-1})}{(1-\gamma)}.
\end{align*}

So we have,
\begin{align*}
    \| \nabla_{\Sigma} f(K',\Sigma') - \nabla_{\Sigma}f(K, \Sigma) \| & \leq \frac{\gamma (B^\top B +D^\top D) }{(1-\gamma)} |P_{K'} - P_K|+  \frac{\tau}{2(1-\gamma)}\|\Sigma'^{-1}-\Sigma^{-1}\|\\
    & \leq  \frac{\gamma (B^\top B +D^\top D) }{(1-\gamma)} h_5 \|K'-K\|+  \frac{\tau\sigma_{min}(\Sigma)}{4(1-\gamma)}\|\Sigma-\Sigma'\|\\
     & \leq  \frac{\gamma (B^\top B +D^\top D) }{(1-\gamma)} h_5 \|K'-K\|+  \frac{\tau\sigma_{min}(\Sigma)}{4(1-\gamma)}\|\Sigma'-\Sigma\|_F\\
    &:= h_8\|K'-K\| + h_9\|\Sigma'-\Sigma\|_F.
\end{align*}
The proof is completed
\end{proof}

\subsection{Proof of Lemma \ref{lemma: fr(K)}}
\begin{proof}
The proof is provided in Lemma 2.1 of \cite{Flaxman2005online} with slightly change of notation.
\end{proof}

\subsection{Proof of Lemma \ref{lemma: est gk}}
\begin{proof}
We first show that finitely many finite-horizon rollouts, defined as 
$$
\widetilde{\nabla}_K := \frac{1}{M} \sum_{i=1}^M \frac{n}{r^2} f(K + U_i, \Sigma) U_{i}.
$$
If $r_1$ satisfies the conditions in the lemma, then $\|\widetilde{\nabla}_K -\nabla_K f(K,\Sigma)\|_F < \frac{\epsilon}{3}$ with high probability(at least $1-\sqrt{\kappa_1}$). We break $\widetilde{\nabla}_K -\nabla_K f(K,\Sigma)$ into the following two parts,
$$
\widetilde{\nabla}_K -\nabla_K f(K,\Sigma) = (\nabla_K f_{r_1}(K,\Sigma) - \nabla_K f(K, \Sigma)) + (\widetilde{\nabla}_K  - \nabla_K f_{r_1}(K, \Sigma)).
$$
For the first term, from Lemma \ref{lemma: gradient perturbation} we have $ \| \nabla_K f(K + U, \Sigma) - \nabla_K f(K, \Sigma) \| \leq \frac{\epsilon}{6}$ if we set $r_1 \leq \frac{\epsilon}{6 h_6}$. 
Since $\nabla_K f_{r_1}(K, \Sigma)$ is the expectation of $\nabla f(K+U, \Sigma)$, we have $\|\nabla_K f_{r_1}(K, \Sigma) - \nabla_K f(K, \Sigma)\|_F \leq \frac{\epsilon}{6}$.

For the second term, $\widetilde{\nabla}_K  - \nabla_K f_{r_1}(K, \Sigma)$. We want to invoke the Vector Bernstein Inequality to show that with high probability $||\widetilde{\nabla}_K - \nabla_K f_{r_1}(K, \Sigma)|| < \frac{\epsilon}{2}$. Consider the sample $i^{th}$ of a single path $K+U_i$ and observe that $\|\frac{n}{r_1^2}f(K+U_i)U_i \| \leq \frac{2nf(K, \Sigma)}{r_1}$ if we assume $f(K+U_i, \Sigma) \leq 2f(K, \Sigma)$. Also, 
\begin{align*}
     \| \nabla_K f_{r_1}(K, \Sigma) \| 
     &\leq \|\nabla_K f_{r_1}(K, \Sigma) - \nabla_K f(K, \Sigma)\|_F + \nabla_K f(K,\Sigma)\\ 
     &\leq \frac{\epsilon}{6} + \overline{\| \nabla_K f(K, \Sigma) \|}.
\end{align*} 
So we have 
\begin{align*}
    \| \frac{n}{r_1^2} f_{r_1}(K+U_i, \Sigma)U_i - \nabla_K f_{r_1}(K, \Sigma) \| \leq R_1:= \frac{2n}{r_1}f(K, \Sigma) + \frac{\epsilon}{6} + \overline{\| \nabla_K f(K, \Sigma) \|}
\end{align*} 
and 
\begin{align*}
     &\| \mathbb{E}\bigg[\big(\frac{n}{r_1^2} f_{r_1}(K+U_i, \Sigma)U_i\big)^\top\frac{n}{r_1^2} f_{r_1}(K+U_i, \Sigma)U_i\bigg] - \nabla^\top_K f_{r_1}(K, \Sigma)\nabla_K f_{r_1}(K, \Sigma) \|\\
     &\quad  \leq \max_{U_i}\|\frac{n}{r_1^2} f_{r_1}(K+U_i, \Sigma)U_i\|_F^2 + \|\nabla_K f_{r_1}(K, \Sigma)\|^2_F\\
     &\quad \leq \sigma_1 := \big(\frac{2n}{r_1}f(K, \Sigma)\big)^2 + \big(\frac{\epsilon}{6} + \overline{\| \nabla_K f(K, \Sigma) \|}\big)^2.
\end{align*}
Next, note that $ \mathbb{E}{[\frac{n}{r_1^2}f(K+U_i, \Sigma)U_i]}=\mathbb{E}{[\hat{\nabla}_K]} = \nabla_K f_{r_1}(K, \Sigma)$,  apply Vector Bernstein Inequality we have if 
$$
M \geq \frac{2n}{(\epsilon_1/6)^2}( \sigma_\nabla + \frac{R_\nabla}{3\sqrt{n}})\log\big(\frac{n+1}{\sqrt{\kappa_1}}\big)
$$
then 
$$
\mathbb{P}\big[ \|\widetilde{\nabla}_K  - \nabla_K f_{r_1}(K, \Sigma)\|_F \leq\frac{\epsilon}{6}\big] \geq 1- \sqrt{\kappa_1}.
$$

In the above, we have demonstrated that $\|\widetilde \nabla_K  - \nabla_K f(K,\Sigma)\|_F \leq \frac{\epsilon}{3}$ with high probabiliy. Now we attempt to prove that $\widehat{\nabla}_K$ is $\epsilon$ close to $\nabla_K f(K,\Sigma)$ with high probability $1-\kappa_1$, under the conditions given in the lemma. Define 
$$
\nabla_K' := \frac{1}{M} \sum_{i=1}^M \frac{n}{r^2_1} f^{(l)}(K+U_i, \Sigma)U_{i}.
$$
We break $\widehat{\nabla}_K - \nabla_K f(K,\Sigma)$ into following three parts,
$$
\widehat{\nabla}_K - \nabla_K f(K, \Sigma) = (\widehat{\nabla}_K - \nabla_K') + (\nabla_K'  - \widetilde{\nabla}_K) +  (\widetilde{\nabla}_K - \nabla_K f(K, \Sigma)).
$$ 

For third term, based on the previous proof, we know that $\|\widetilde{\nabla}_K - \nabla_K f(K, \Sigma) \| \leq \frac{\epsilon}{3}$ with probability at least $1-\kappa_1$ under the conditions given in the lemma.  

For the second term, using the lemma \ref{lemma: approximate fk} we have 
\begin{align*}
    \|\nabla_K'  - \widehat{\nabla}_K \| & = \frac{1}{M} \cdot \frac{n}{r_1^2} \bigg| \bigg| \sum_{i = 1}^{M} \left(f^{(l)}(K+U_i, \Sigma) - f(K+U_i, \Sigma) \right) U_i  \bigg| \bigg| \notag\\
    & \leq \frac{1}{M} \cdot \frac{n}{r_1^2} \bigg| \bigg| \sum_{i= 1}^M \gamma^l\left[ (Q + (K+U_i)^\top R(K+U_i))S_{K+ U_i,\Sigma} + \frac{ Tr(\Sigma R) - \frac{\tau}{2}(n+ log(2\pi)^n|\Sigma|)}{1-\gamma }\right] U_i \bigg| \bigg| \notag\\
    & \leq \frac{1}{M} \cdot \frac{n}{r_1^2} \left( \sum_{i= 1}^M \gamma^l\left[ (|Q| \cdot |S_{K+U_i, \Sigma}| + \|K+U_i\|^2 \|R\|\cdot |S_{K+ U_i,\Sigma}| + \Bigg|\frac{ Tr(\Sigma R) - \frac{\tau}{2}(n+ log(2\pi)^n|\Sigma|)}{1-\gamma } \Bigg| \right] r_1 \right) \notag\\
    & \leq \frac{1}{M} \cdot \frac{n}{r_1} \left( M \gamma^l\left[ |2f(K, \Sigma)-\psi| + 2\|K\|^2 \|R\|\cdot \Bigg| \frac{2f(K, \Sigma)- \psi}{Q} \Bigg| + \Bigg|\frac{\psi}{1-\gamma } \Bigg| \right] \right) \notag\\
    & \leq \frac{n}{r_1} \left( \gamma^l\left[ |2f(K, \Sigma)|+|\psi| + 2\|K\|^2 \|R\|\cdot \frac{|2f(K, \Sigma)|+ |\psi|}{|Q|} \Bigg| + \Bigg|\frac{\psi}{1-\gamma } \Bigg| \right] \right) \notag\\
    & \leq \frac{\epsilon}{3},
\end{align*}
if
$
l \geq \log (\gamma)^{-1} \left[\log\left(\frac{r_1}{n} \cdot \frac{\epsilon}{3} \right)- \log \left( 2|f(K, \Sigma)| \left(2 \| K \|^2\|R \| + \frac{1}{|Q|}\right) + |\psi|\left(1 + \frac{1}{|Q|} + \frac{1}{1-\gamma}\right) \right) \right]
$, 
where $\psi = Tr(\Sigma R) - \frac{\tau}{2}(n + \log(2\pi)^n|\Sigma|)$.

For the first term, note that $|x_0^i| \leq L$ and let $\Gamma>1$ such that $$\sum_{j=0}^{l-1} Q(x_t^i)^2 + (u_t^i)^\top R(u_t^i) \leq \Gamma \mathbb{E}\left[ \sum_{i=0}^{l-1} Qx_i^2 + u_t^\top Ru_t \right]$$ Thus, $\widehat{\nabla}_K - \nabla'_K$ is the sum of random bounded vectors. Now observe:
\begin{align}
    \sum_{t=0}^{l-1} Q(x_t^i)^2 + (u_t^i)^\top R(u_t^i) + \tau \log\pi(u_t^i|x_t^i)& \leq \Gamma \mathbb{E}\left[ \sum_{t=0}^{l-1} Qx_t^2 + u_t^\top Ru_t +\tau \log\pi(u_t^i|x_t^i) \right] \notag\\
    & \leq \Gamma \mathbb{E}\left[ \sum_{t=0}^{\infty}
    Qx_t^2 + u_t^\top Ru_t +\tau \log\pi(u_t|x_t) \right] \notag\\
    & \leq \Gamma L^2 f(K + U_i, \Sigma) \notag\\
    & \leq 2 \Gamma L^2 f(K, \Sigma).
\end{align}
And since $$ \| \nabla_K' \| \leq \frac{\epsilon}{3} + \| \widetilde{\nabla}_K \| \leq \frac{\epsilon}{2} + \overline{\| \nabla_K f(K, \Sigma) \|},$$ 
we have 
\begin{align*}
     \| \big[\sum_{t=0}^{l-1} Q(x_t^i)^2 + (u_t^i)^\top R(u_t^i) + \tau \log\pi(u_t^i|x_t^i) \big]U_i - \nabla_K'\|_F & \leq 2\Gamma L^2 f(K, \Sigma)\|U_i\|_F  + \|\nabla_K'\|_F\\
     & \leq R_{2}  \\
     &:= 2\Gamma L^2 f(K, \Sigma)r_1 + \frac{\epsilon}{2} + \overline{\| \nabla_K f(K, \Sigma) \|}.
\end{align*}
Define $Z_i := \big[\sum_{t=0}^{l-1} Q(x_t^i)^2 + (u_t^i)^\top R(u_t^i) + \tau \log\pi(u_t^i|x_t^i) \big]U_i$.
\begin{align*}
    \|\mathbb{E}[ Z_i^\top Z_i ] - (\nabla_K')^\top \nabla_K'\big]\| &\leq \max_{U_i} \|Z_i\|_F^2 + \|\nabla_K'\|_F^2\\
    &\leq \sigma_{2}\\
    &:= (2\Gamma L^2 f(K, \Sigma)r_1)^2 + (\frac{\epsilon}{2} + \overline{\| \nabla_K f(K, \Sigma) \|})^2.
\end{align*}
the norm of each sample is bounded, assuming the variance is bounded, and since  $\mathbb{E}[\widehat{\nabla}_K] = \nabla'_K$.
Then by the Vector Bernstein Inequality we have 
if 
$$
M\geq \frac{2n}{(\epsilon/3)^2}(\sigma^2_{2} + \frac{R_{2} \epsilon}{3\sqrt{n}})\log(\frac{n+1}{\sqrt{\kappa_1}}) 
$$
then 
\begin{align*}
    \mathbb{P} \left( \|\widehat{\nabla}_K - \nabla_K'\| \leq \frac{\epsilon}{3} \right) \leq 1 - \sqrt{\kappa_1}.
\end{align*}  
Combining the above we have 
$$
\mathbb{P} \left( \|\widehat{\nabla}_K - \nabla_Kf(K,\Sigma)\| \leq{\epsilon} \right) \leq 1 - \kappa_1.
$$
\end{proof}

\subsection{Proof of Lemma \ref{lemma: est gsigma}}
\begin{proof}
	Define: 
	$$
		\widehat{\nabla}_L  := \frac{1}{M} \sum_{i=1}^M \frac{n(n+1)}{2r_2^2} \left[\sum_{t = 0}^{l-1} \gamma^t\left(Q(x_t^i )^2 + (u_t^i)^\top Ru_t^i  + \tau \log \pi (u_t|x_t)\right)\right] V_{i}.
	$$

We first show that finitely many finite-horizon rollouts, defined as
$$
\widetilde{\nabla}_L := \frac{1}{M} \sum_{i=1}^M \frac{n(n+1)}{2r_2^2} f(K, \Sigma_i) V_{i},
$$
$ \|\widetilde{\nabla}_L- \nabla_L f(K, \Sigma)\| \leq \frac{\epsilon}{3}$   with high probability under some conditions. We break $ \widetilde{\nabla}_L - \nabla_L f(K, \Sigma)$ into two terms,
\begin{align}
	\widetilde{\nabla}_L - \nabla_L f(K, \Sigma) & =  ( \nabla_L f_{r_2}(K, L) - \nabla_L f(K, \Sigma))+ ( \widetilde{\nabla}_L -\nabla_L f_{r_2}(K, L) ).
\end{align}
For the first term, by Lemma $4.5$ we have $ \|\nabla_L f(K, (L + V)(L+V)^\top) - \nabla_L f(K, \Sigma)\| < \frac{\epsilon}{6}$ if $r_2 \leq \frac{\epsilon}{ 6h_9}$. And because $\nabla_L f_{r_2}(K, L) = \mathbb{E}[\nabla_L f(K, (L + V)(L+V)^\top)]$, we have $$\| \nabla_L f_{r_2}(K, L) - \nabla_L f(K, \Sigma)\|_F \leq \frac{\epsilon}{6}$$ by triangle inequality.

For the second term $\widetilde\nabla_{L} -\nabla_L f_{r_2}(K, L)$. Note that $ \mathbb{E} \left[ \widetilde{\nabla}_L \right] = \sum_{i=1}^{M} \|\frac{2n(n+1)}{2r_2^2} f(K, \Sigma) V_{i}\| = \nabla_L f_{r_2}(K, L)$, using $f(K,\Sigma_i) \leq 2 f(K,\Sigma)$ we have 
$$\| \frac{n(n+1)}{2r_2^2} f(K, \Sigma_i) V_{i} \|\leq \|\frac{2n(n+1)}{2r_2^2} f(K, \Sigma) V_{i}\| \leq \frac{n(n+1)}{r_2} f(K, \Sigma),$$
using the the result of the first term and triangle inequality we have 
$$
 \|\nabla_L f_{r_2}(K, L)\|_F \leq \|\nabla_L f (K, \Sigma)\|_F + \frac{\epsilon}{6}.
$$
$$
\| \frac{n(n+1)}{2r_2^2} f(K, \Sigma_i) V_{i}  - \nabla_L f_{r_2}(K, L)\|_F \leq \frac{n(n+1)}{r_2} f(K, \Sigma) + \overline {\|\nabla_L f(K, \Sigma)\|} + \frac{\epsilon}{6}.
$$
so we have
$$
\mathbb{E}[ \|\frac{n(n+1)}{2r_2^2}f(K, \Sigma_i)V_i -\nabla_L f_{r_2}(K, \Sigma)\|_F^2 ] \leq  R_{Z_1} :=\frac{n(n+1)}{r_2} f(K, \Sigma) + \overline {\|\nabla_L f_{r_2}(K, L)\|} + \frac{\epsilon}{6}.
$$
Next we consider 

\begin{align*}
& \mathbb{E} ((\frac{n(n+1)}{2r_2^2}f(K, \Sigma_i)V_i)^{\top} \frac{n(n+1)}{2r_2^2}f(K, \Sigma_i)V_i )) - \nabla_L f_{r_2}(K,L)^{\top}\nabla_L f_{r_2}(K,L)\\
& \quad \leq \max_{V_i} \|\frac{n(n+1)}{2r_2^2}f(K, \Sigma_i)V_i\|_F - \|\nabla_L f_{r_2}(K,L)\|_F^2 \\
& \quad \leq  \sigma_{Z_1} := \frac{n(n+1)}{r_2} f(K, \Sigma),
\end{align*}
Using Vector Bernstein Inequality,  we know that if 
$$
M \geq \frac{2}{(\epsilon/3)^2}(\sigma_{Z_1}^2 + \frac{R_{Z_1}{\epsilon/3}}{3}) \log(\frac{1}{1-\sqrt{1-\kappa}})
$$ 
then 
$
 \|\widetilde{\nabla}_L- \nabla_L f(K, \Sigma)\| \leq \frac{\epsilon}{3}
$
with probability at least $1-\kappa$.

In the above, we have demonstrated that $\widetilde \nabla_\Sigma$ is $\epsilon$ close to $\nabla_\Sigma f(K,\Sigma)$ under some conditions. Now we attempt to prove that $\widehat{\nabla}_K$ is $\epsilon$ close to $\nabla_\Sigma f(K,\Sigma)$ with high probability $1-\kappa_2$, under the conditions given in the lemma. Define
$$ \nabla_L' := \frac{1}{M}\sum_{i=1}^M \frac{n(n+1)}{2r_2^2} f^{(l)}(K, \Sigma_i)V_i. $$ 
We break $\widehat{\nabla}_L - \nabla_L f(K, \Sigma)$ into three terms, 
$$
\widehat{\nabla}_L - \nabla_L f(K, \Sigma) = (\widehat{\nabla}_L - \nabla_L') + (\nabla_L'  - \widetilde{\nabla}_L) +  (\widetilde{\nabla}_L - \nabla_L f(K, \Sigma)).
$$

For the third term, based on the previous proof, we have $ \| \widetilde{\nabla}_L - \nabla_L f(K, \Sigma) \| \leq \frac{\epsilon}{3}$ with probability $1-\kappa_2$ under conditions given in the lemma. Furthermore, we have  $ \|\widetilde{\nabla}_L\| \leq \|\nabla_L f(K, \Sigma) \| + \frac{\epsilon}{3}.$

For the second term, we have 
\begin{align*}
	\| \nabla_L'  - \widetilde{\nabla}_L \|&=   \frac{n(n+1)}{2r_2^2} \frac{1}{M} \sum_{i = 1}^{M} \left(f^{(l)}(K, \Sigma_i) - f(K, \Sigma_i) \right) \|V_i\|\\
	& \leq  \frac{n(n+1)}{2r_2} \frac{1}{M} \sum_{i = 1}^{M} \left(f^{(l)}(K, \Sigma_i) - f(K, \Sigma_i) \right) .
\end{align*}
If $$
l\geq \frac{\log \frac{2r_2\epsilon}{3n(n+1)} - \log \left[(Q + K^\top RK)S_{K,\Sigma} + \frac{ Tr(\Sigma R) - \frac{\tau}{2}(n+ log(2\pi)^n|\Sigma|)}{1-\gamma }\right]}{log \gamma},
$$
 then $\| \nabla_L'  - \widetilde{\nabla}_L \| \leq \frac{\epsilon}{3}$ by  lemma  4.2. Furthermore, by triangle inequality we have 
 $$ \|\nabla_L'\|  \leq \|\widetilde{\nabla}_L\|  +  \frac{\epsilon}{3} \leq \|\nabla_L f(K, \Sigma) \| + \frac{2\epsilon}{3} \leq \overline{\| \nabla_L f(K, \Sigma) \| }+ \frac{2\epsilon}{3}.
 $$
 
 For the first term, we will use Vector Bernstein inequality as $\mathbb{E}[\widehat{\nabla}_L] =  \nabla'_L$. By the assumption we have
\begin{align*}
	\sum_{t=0}^{ l-1} Q(x_t^i)^2 + (u_t^i)^\top R(u_t^i) + \tau \log\pi(u_t|x_t)& \leq \Gamma \mathbb{E}\left[ \sum_{t=0}^{l-1} Q(x_t^i)^2 + (u_t^i)^\top Ru_t^i +\tau \log\pi(u_t^i|x_t^i) \right] \notag\\
	& \leq \Gamma \mathbb{E}\left[ \sum_{t=0}^{\infty}
	Q(x_t^i)^2 + (u_t^i)^\top Ru_t^i + \tau \log\pi(u_t^i|x_t^i) \right] \notag\\
	& \leq \Gamma L^2 f(K,\Sigma_i) \notag\\
	& \leq 2 \Gamma L^2 f(K, \Sigma).
\end{align*}
where $|x_0^i| \leq L$.
Define $\widehat{\nabla}^i_L$ so that $\widehat{\nabla}_L = \sum_{i=1}^{M} \widehat{\nabla}^i_L$
\begin{align*}
\|\widehat{\nabla}^i_L - {\nabla}_L'\|_F & \leq \|\widehat{\nabla}^i_L\|_F + \| {\nabla}_L' \|_F\\
& \leq \frac{n(n+1)}{r_2} \Gamma L^2 f(K, \Sigma) + \| {\nabla}_L' \|_F\\
& \leq R_Z := \frac{n(n+1)}{r_2} \Gamma L^2 f(K, \Sigma)  + \overline{\| \nabla_L f(K, \Sigma) \| }+ \frac{2\epsilon}{3}.
\end{align*}
\begin{align*}
& \mathbb{E} [(\widehat{\nabla}^i_L)^{\top} \widehat{\nabla}^i_L ]- {\nabla_L '}^{\top}\nabla_L'\\
& \quad \leq \max_{V_i} \|\widehat{\nabla}^i_L \|^2_F - \|\nabla_L'\|_F^2 \\
& \quad \leq \|\frac{n(n+1)}{2r_2^2} 2 \Gamma L^2 f(K, \Sigma) V_i\|^2 - \|\nabla_L'\|_F^2\\
& \quad \leq \big(\frac{n(n+1)}{r_2} \Gamma L^2 f(K, \Sigma)\big)^2 - \|\nabla_L'\|_F^2\\
& \quad \leq  \sigma_{Z_2} :=  \big(\frac{n(n+1)}{r_2} \Gamma L^2 f(K, \Sigma)\big)^2.
\end{align*}

If 
$$
M > \frac{2n(n+1)}{2(\epsilon/3)^2} \big(\sigma_{{Z_2}}^2 +\frac{R_{{Z_2}} \epsilon}{9\sqrt{n(n+1)/2}}\big)\log\frac{n(n+1)/2 + 1}{1-\sqrt{1-\kappa}}
$$
then 
$$
\|\widehat{\nabla}_L - \nabla'_L \|<\frac{\epsilon}{3}
$$ 
with probability $1-\sqrt{1-\kappa}$.
Combining the above we have 
$$
\|\widehat{\nabla}_L - \nabla_L f(K,\Sigma)\| \leq \epsilon
$$
with probability at least $1-\kappa$.
As 
\begin{align*}
	\nabla_{L} f(K,\Sigma = LL^\top) &= \nabla_{L} (LL^\top) \nabla_{\Sigma} f(K,\Sigma)\\
	& = 2L \nabla_{\Sigma} f(K,\Sigma),
\end{align*}
Now take $\epsilon = \frac{\epsilon_2}{2\|L^{-1}\|}$, we have
$$
\|\widehat{\nabla}_\Sigma - \nabla_\Sigma f(K,\Sigma)\|  = \|(2L)^{-1}\widehat \nabla_{\Sigma} - (2L)^{-1}\nabla_{\Sigma} f(K,\Sigma)\| \leq \epsilon_2
$$
with probability $1-\kappa$.

\end{proof}

\subsection{Proof of Lemma \ref{lemma: approxi S under perturbation}}
\begin{proof}
Note that if $r = \min\{\frac{S_{K,\Sigma}}{2h_K}, h_\Sigma\}$, we have $ \frac{S_{K,\Sigma}}{2}\leq S_{K+U_i, \Sigma} \leq \frac{3S_{K,\Sigma}}{2}$.
Let $\widetilde S  = \frac{1}{M}\sum_{i=1}^{M} S_{K+U_i, \Sigma}$ and $\widetilde S^{(l)} = \frac{1}{M}\sum_{i=1}^{M} S^{(l)}_{K+U_i, \Sigma}$
We broke $\widehat S - S_{K,\Sigma}$ into the following three terms:

$$
\widehat S - {S_{K,\Sigma}} =\widehat S -\widetilde S^{(l)} + \widetilde S^{(l)} -\widetilde S  + \widetilde S    -\ {S_{K,\Sigma}} .
$$
For the first term $\widehat S - \widetilde S^{(l)}$, we have $ \mathbb{E}[S_i^{(l)} ] = S^{(l)}_{K+U_i, \Sigma}$. Apply Bernstein we have$
|\widehat S - \widetilde S^{(l)}| = \frac{1}{M}\sum_{i=0}^M |S_{i}^{(l)} - S^{(l)}_{K+U_i, \Sigma}| \leq \frac{\epsilon}{3} $ with probability at least $1 - n \exp\{-\frac{M^2\epsilon}{ 3S_{K,\Sigma} }\}$.
For the second term $\widetilde S^{(l)} -\widetilde S  $,  set $l \geq \frac{\log \epsilon/3 - \log S_{K,\Sigma}/2}{log \gamma}$ and apply lemma \ref{lemma: approximate fk}, we have $\|\widetilde S^{(l)} -\widetilde S  \|\leq \frac{\epsilon}{3}$.
For the third term  $\widetilde S   -\ {S_{K,\Sigma}}$, apply lemma \ref{lemma: sk perturbation} we have $\|S_{K+U_i,\Sigma} - S_{K,\Sigma}\| \leq \frac{\epsilon}{3}$ if $\|U_i\|_F \leq \min\{\frac{\epsilon}{3h_K},h_\Sigma\}$. As $\widetilde S $ is the average of $S_{K+U_i,\Sigma}$, we have $\|\hat S_{K,\Sigma}  -\ {S_{K,\Sigma}}\| \leq \frac{\epsilon}{3}$.

For the bound on $\widehat S $, apply Weyl’s Theorem when $\epsilon \leq \mu/2$ we have $\sigma_{min}(\widehat{S}) \geq \sigma_{min}(S_{K,\Sigma}) - \mu/2 \geq \mu/2 $
\end{proof}

\subsection{Proof of Lemma \ref{lemma: f perturbation}}

\begin{proof} By triangle inequality we have:
\begin{align*}
    |q_{K',\Sigma'} - q_{K,\Sigma}| \leq  |q_{K',\Sigma'} - q_{K,\Sigma'}| 
    + |q_{K,\Sigma'} -q_{K,\Sigma}|.
\end{align*}
For the first term,  when $\|K'-K\| \leq \min\{h_\Sigma, \|K\|\}$ and $\|\Sigma'-\Sigma\| \leq \|\Sigma\|$, using lemma \ref{lemma: pk perturbation} we have  
\begin{align*}
    |q_{K',\Sigma'} - q_{K,\Sigma'}| &= \frac{ \text{Tr}(\Sigma' \gamma (P_K-P_K') (B^\top B+D^\top D))}{1-\gamma}\\
    &\leq \frac{\gamma\|\Sigma'\|  \|B^\top B+D^\top D\|}{1-\gamma}|P_K-P_{K'}|\\
    &\leq \frac{2\gamma\|\Sigma\|  \|B^\top B+D^\top D\|}{1-\gamma} h_5\|K'-K\|\\
    &:= h_{10} \|K'-K\|.
\end{align*}
For the second term, from intermediate step of lemma \ref{lemma: almost smooth}, we have:
\begin{align*}
    |q_{K,\Sigma'}-q_{K,\Sigma}| &\leq \frac{m}{2}Tr((\Sigma^{-1}\Sigma'-I)^2) + Tr\left(\nabla_{\Sigma}q_{K,\Sigma}^\top (\Sigma-\Sigma')\right)\\
    &\leq \frac{m}{2} \|\Sigma^{-1}\Sigma'-I\|_F^2 +\|\nabla_{\Sigma}q_{K,\Sigma}\| \|\Sigma'-\Sigma\|\\
    & \leq \frac{m\|\Sigma^{-1}\|_F}{2} \|\Sigma'-\Sigma\|_F^2 + \overline{\|\nabla_{\Sigma}q_{K,\Sigma}\|} \|\Sigma'-\Sigma\|_F\\
    &:= h_{11} \|\Sigma'-\Sigma\|_F.
\end{align*}
Combining the above completes the proof.
\end{proof}

\section{Standard Matrix Perturbation and Concentrations}
In this section, we review several basic matrix tools that were used throughout the paper.
\subsection{Vector Bernstein inequality}\begin{lemma}
Let $\{ Z_i \}_{i=1}^N$ be a set of $N$ independent random vectors of dimension ${n}$ with ${\mathbb{E} [Z_i] = Z}$, ${\| Z_i - Z\| \leq R_Z}$ almost surely, and maximum variance 
${\left\| \mathbb{E} ( Z_i^\intercal Z_i) - Z^\intercal Z\right\|\leq \sigma_Z^2}$, and sample average ${\widehat{Z} := \frac{1}{N} \sum_{i=1}^N Z_i}$.
Let a small tolerance ${\epsilon \geq 0}$ and small probability ${ 0 \leq \kappa \leq 1 }$ be given.
If
\begin{align*}
    N \geq \frac{2n}{\epsilon^2} \left( \sigma_Z^2 + \frac{R_Z \epsilon}{3\sqrt{n}}\right) \log \left[ \frac{n}{\mu} \right]
\end{align*}
then
\begin{align}
    \mathbb{P} \left[ \left\| \widehat{Z} - Z \right\|_F \leq \epsilon \right] \geq 1-\mu .
\end{align}
\end{lemma}
\begin{proof}
This Lemma is directly obtained by applying Lemma C.6 in \cite{Gravell2019robust} to the case of vectors.
\end{proof}

\subsection{Weyl's Inequality for singular values}
Suppose $B=A+E$, then the singular values of $B$ are within E to the
 corresponding singular values of A. In particular, $\|B\| \leq \|A\|+\|E\|$ and $\sigma_{min}(B) \geq \sigma_{min}(A) - \|E\|$.

\subsection{Perturbation of Inverse}
Let $B=A+E$, suppose $E \leq \sigma_{min}(A) /2$, then $\|B^{-1} - A^{-1}\| \leq 2 \|A - B\| / \sigma_{min} (A)$.

\subsection{Matrix Norm}
For matrix $A,B\in \mathbb R^{n\times m}$, we have
 $   \|A^{-1}\| \geq \|A\|^{-1}$  and 
    $|Tr(A^\top B)| \leq \|A^\top \||Tr(B)| = \|A\||Tr(B)|$.
If $A \succeq 0$, we have
$Tr(A) \geq \|A\| $.

\end{document}